\documentclass[sigconf]{acmart}

\newtheorem{remark}{Remark}[section]
\newtheorem{assumption}{Assumption}[section]
\usepackage{subfig}
\usepackage{algorithm}
\usepackage{algorithmic}

\AtBeginDocument{%
  }

\copyrightyear{2025}
\acmYear{2025}
\setcopyright{cc}
\setcctype{by}
\acmConference[MobiHoc '25]{The Twenty-sixth International Symposium on
Theory, Algorithmic Foundations, and Protocol Design for Mobile Networks and
Mobile Computing}{October 27--30, 2025}{Houston, TX, USA}
\acmBooktitle{The Twenty-sixth International Symposium on Theory,
Algorithmic Foundations, and Protocol Design for Mobile Networks and Mobile
Computing (MobiHoc '25), October 27--30, 2025, Houston, TX, USA}
\acmDOI{10.1145/3704413.3764437}
\acmISBN{979-8-4007-1353-8/2025/10}

\usepackage{xspace}
\newcommand{\system}{\textit{Odin}\xspace}
\begin{document}

\title{Odin: Effective End-to-End SLA Decomposition for 5G/6G Network Slicing via Online Learning}

\author{
Duo Cheng\textsuperscript{1}, 
Ramanujan K Sheshadri\textsuperscript{2},
Ahan Kak\textsuperscript{2},
Nakjung Choi\textsuperscript{2},
Xingyu Zhou\textsuperscript{3},
Bo Ji\textsuperscript{1}
}
\affiliation{
\textsuperscript{1} Virginia Tech  \quad
\textsuperscript{2} Nokia Bell Labs \quad
\textsuperscript{3} Wayne State University \quad \\
\country{\{duocheng, boji\}@vt.edu\quad \{ram.sheshadri, ahan.kak, nakjung.choi\}@nokia-bell-labs.com  \quad
xingyu.zhou@wayne.edu}
}

\settopmatter{printacmref=false} 
\renewcommand\footnotetextcopyrightpermission[1]{} 
\pagestyle{plain}
\settopmatter{printfolios=true}

\begin{abstract}
Network slicing plays a crucial role in realizing 5G/6G advances, enabling diverse Service Level Agreement (SLA)
requirements related to latency, throughput, and reliability. Since network slices are deployed
end-to-end (E2E), across multiple domains including access, transport, and core networks, it is essential to
efficiently decompose an E2E SLA into domain-level targets, so that each domain can provision adequate resources for the slice. However, decomposing SLAs is highly challenging due to the heterogeneity of domains, dynamic network conditions, and the fact that the SLA orchestrator
is oblivious to the domain's resource optimization. In this work, we
propose \system, a Bayesian Optimization-based solution that leverages each domain's online feedback for
\textit{provably-efficient} SLA decomposition. Through theoretical analyses and rigorous evaluations, we demonstrate that \system's E2E orchestrator can achieve up to 45\% performance improvement in SLA satisfaction when compared with baseline solutions whilst reducing overall resource costs even in the presence of noisy feedback from the individual domains.
\end{abstract}

\maketitle

\section{Introduction}
\label{sec:intro}
The advent of 5G/6G has enabled diverse services with stringent Service Level Agreement (SLA) requirements. SLAs, defined by metrics such as latency, throughput, and reliability, set end-to-end (E2E) performance expectations across core, transport, and access (e.g., Radio Access Network, RAN) domains. To meet these needs, network slicing creates multiple virtual networks over shared infrastructure, each tailored to specific SLAs, ensuring efficient implementation network-wide.

\noindent \emph{\bf End-to-End SLA decomposition:}
E2E SLA decomposition~\cite{de2020decomposing,abbasi2013online,iannelli2020applying,hsu2024online} is a crucial step in realizing SLA demands, 
as it enables partitioning E2E service requirements into domain-specific targets, facilitating effective resource allocation. 
However, SLA decomposition is challenging as the core, transport, and RAN domains have distinct resource needs, technologies, and operational characteristics. 
Moreover, as networks evolve, domain orchestrators 
continuously adjust their resource allocations, potentially leading to performance variations that may hinder meeting SLA targets.

Fig.~\ref{fig:motiv} illustrates how two different decompositions of a slice's SLA, which requires an E2E latency at most 100~ms, can yield very different outcomes. In the first, 
the RAN with 20 Physical Resource Blocks (PRBs) can achieve a latency of 53~ms, well above the decomposition target of 10~ms. While the transport and
core networks can easily meet their respective latency targets of 60~ms and 30~ms, 
their resources (link bandwidth of 1~Gbps and CPU capacity of 5~GHz, respectively) remain underutilized. This suboptimal decomposition
results in an E2E latency 
of 143~ms (SLA satisfaction of 70\%). In contrast, 
the second decomposition, which sets smaller but feasible targets for the transport and core networks (38~ms and 2~ms, respectively) compared to the RAN (60~ms), is able to 
achieve 100\% SLA satisfaction through better allocation across domains.

\begin{figure}[!t]
	\centering
\includegraphics[width=1.02\linewidth]{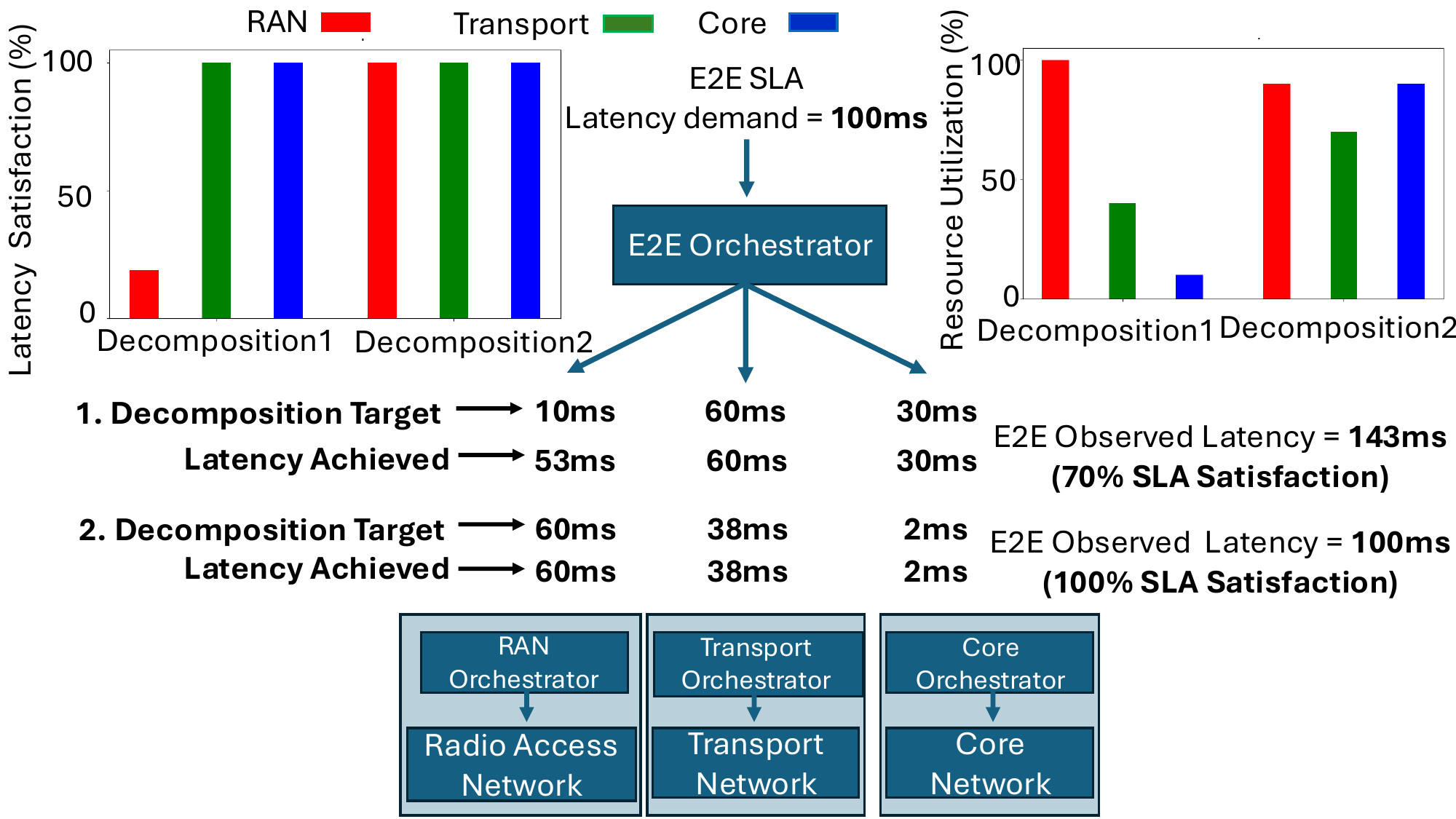}
	\caption{Impact of SLA decomposition on E2E performance}
    
    \label{fig:motiv}
\end{figure}

Existing work on SLA decomposition~\cite{iannelli2020applying,de2020decomposing,hsu2023sla} has focused on an offline setting, training models on precollected data to approximate domain behaviors. Such models risk failure due to the distribution discrepancy between offline training and online deployment. In contrast, 
our work leverages a novel Bayesian Optimization (BO)-based algorithm to address this issue in an online setting, where the algorithm continuously learns how to make better decomposition decisions through low-overhead
interactions with the network, without prior knowledge of domain-specific conditions or control policies, which, however, faces the following challenges.

\subsection{Challenges}
\textbf{Unknown network conditions:}
The primary challenge for SLA decomposition stems from \emph{dynamic} network conditions. As network conditions vary, so does the domain's performance.
Decomposing an E2E SLA without considering these fluctuations can be suboptimal. Since domain-specific conditions and resource optimizations are independent of and unknown to the E2E orchestrator, it must continuously learn and adapt SLA decompositions across domains. Achieving this efficiently at scale—across thousands of slices and flows—poses a significant challenge due to potential overhead.

\noindent \textbf{Modeling domain performance:} The orchestrator responsible for resource optimization within each domain
can potentially use different combinations of resources to achieve the required performance. E.g., the RAN
orchestrator can adjust a wide range of resources such as power control, handover thresholds, 
interference management, scheduling weights, antenna tilt, and PRBs, with each resource having varying 
degree of impact on the RAN performance. Resource choices may change over time due to operator preferences, vendor architectures, or monetary constraints. Thus, it is generally infeasible for E2E orchestrators to access or explicitly use domain-level control details, as any change from domain orchestration in this case can render the E2E orchestrator incompatible.

\subsection{Contributions}
To address the above challenges, we propose \system,\footnote{Just as Odin from Norse mythology oversees and connects the different realms, the E2E orchestrator manages E2E SLAs by coordinating all the domains from above.}
a novel Bayesian Optimization-based solution as a provably-efficient algorithm for
SLA decomposition. \system is 
\emph{agnostic} to the domains' own conditions and control contexts. As we show later in Sec.~\ref{sec:eval}, \system provides  
better empirical performance than other baselines. To that end, we discuss below, the main contributions of this work.

\noindent\textbf{Theoretical formulation with provably-efficient algorithm design:}
To our knowledge, this is the first work to rigorously formulate SLA decomposition as an (constrained) online learning problem. With well-defined objectives, learning algorithm aims to jointly optimize the total resource cost and E2E SLA violation with a delicate balance in between. To handle unknown network conditions and noisy feedback, we adopt Bayesian Optimization (BO) as the learning algorithm and prove its theoretical guarantee. Rather than simply learning the E2E behavior of entire network, \system is tailored to leverage individual domain-level performance, resulting in a more effective and practical solution.

\noindent \textbf{3GPP- and UNEXT-compliant design: }
\system's system design adheres to 3GPP specifications and enables seamless integration into
existing 5G deployments. Additionally, \system's design embodies the emerging UNEXT~\cite{unext} architecture 
in which various heterogeneous modules of the system can autonomously manage themselves while effortlessly interacting with one another via standard interfaces. Furthermore, the optimizations discussed in this paper can be implemented as part of the E2E orchestrator in 
Network Slice Monitoring Function and can leverage standardized REST APIs 
to communicate between the E2E orchestrator and domain orchestrators.

\noindent {\bf{Testbed evaluation and simulation: }}
To evaluate \system, we conduct a detailed trace-based evaluation followed by a large-scale simulation. Trace data are generated using actual core, transport, and RAN testbeds. To that end, we deploy a core network testbed developed using Open5GS~\cite{open5gs}, 
a RAN testbed developed using the OpenAirInterface software stack~\cite{oai,kak2024hexranprogrammableapproachopen}, and a transport testbed using Mininet emulator~\cite{mininet}, which enables deploying virtual networks running on real kernel supporting real-world application traffics over various topologies. As shown in Sec.~\ref{sec:eval}, \system achieves up to 45\% higher SLA satisfaction than various baselines, while significantly reducing overall resource usage across domains.
\section{Background and Related Work}
\label{sec:motivate}
\subsection{Black-Box Optimization}
\label{sec:bbo_bg}
Black-box optimization~\cite{schumer1968adaptive,back1993overview,bergstra2012random,fortin2012deap} (a.k.a. derivative-free optimization) refers to problems where the objective function is unknown or lacks a closed form, and is evaluated via experiments, simulations, or real-world performance.
Without gradient information, Bayesian Optimization—a popular black-box method—uses a surrogate model (e.g., Gaussian Process) to approximate the function with uncertainty quantification, iteratively refining it from observed outputs to improve accuracy.
The goal is to efficiently navigate the search space and to reach the optimal solution. 
Compared to Reinforcement Learning (RL)~\cite{sutton1999reinforcement}, BO is lightweight, requires fewer evaluations to optimize a specific objective, and is more computationally efficient, while RL often needs much more interactions and resources to learn a good solution. In this work, \system treats the domains as black boxes whose performance is learned and modeled by the E2E orchestrator to determine the best SLA decomposition.

\subsection{Related Work}
\label{sec:related}
\textbf{Bayesian Optimization:} Numerous studies have explored the theory of BO. The unconstrained optimization has been studied in \cite{srinivas2010gaussian,chowdhury2017kernelized,scarlett2017lower}. The constrained 
case has been explored in \cite{zhou2022kernelized,deng2022interference}, where a primal-dual approach is proposed with sub-linear guarantees on both regret and ``soft''\footnote{``Soft'' constraint means that it allows a negative violation to cancel out a positive one 
in another round, while ``hard'' constraint does not allow.} constraint violation. \cite{pmlr-v202-xu23h} strengthens this to ``hard'' constraint satisfaction using a different method, which inspires our design. We focus on a specific use case, namely, SLA decomposition, which we tailor the algorithm in \cite{pmlr-v202-xu23h} to.

\noindent \textbf{SLA decomposition:}
There are several previous works that study SLA decomposition, but differ from this work in terms of the problem setup, and/or algorithmic framework. In summary, past solutions are offline and largely assume perfect prior knowledge of network conditions or rely on offline model training, and both of them are unrealistic assumptions for practical deployments.

In \cite{de2020decomposing,hsu2023sla}, the authors formulate a constrained optimization problem where the goal is to find a SLA decomposition that satisfies the E2E SLA, while the probability of all domains accepting the decomposed SLAs is maximized. The authors develop an offline solution in
which they train a surrogate model that is adaptive to domain conditions change. In contrast, our
work employs an online learning algorithm in which the E2E orchestrator determines an SLA decomposition through constant interactions with the domains, re-calibrating the surrogate model to make more accurately model the underlying network conditions. In \cite{hsu2024online}, the formulation in \cite{de2020decomposing,hsu2023sla} is extended to an online scenario, with the new algorithm using a memory buffer to store recent samples for online model update. Besides a different objective from ours, their results are still heuristic without theoretical guarantees.

The authors in \cite{esmat2024sla} target SLA decomposition, but the decision space for the E2E orchestrator includes slice internal configuration, SLA decomposition, resource allocation, etc., while in our formulation, the orchestrator is unaware of (and need not know) the domain orchestration. Since the E2E orchestrator of \system remains \emph{agnostic} to domain operations, it enjoys low communication overhead and is more practical clearly delineating the tasks between the Network Slice Monitoring Function (E2E orchestrator) and Network Slice Subnet Management Functions (domain orchestrator).

The work \cite{hsu2025rails} approaches SLA decomposition by solving a constrained optimization problem with an objective
of maximizing the decomposition-acceptance likelihood, under the constraint that the sum of domain-level latencies falls below the E2E target. However, the decision space includes categorical variables, making the entire optimization problem an NP-hard Mixed-Integer non-linear programming problem, which requires heuristic solutions.
\section{System Design and Problem Formulation
}
\label{sec:formulation}

\subsection{Overview}
The E2E control flow of \system is illustrated in Fig.~\ref{fig:design}. The UNEXT-native E2E orchestrator \cite{unext} implements a black-box optimizer responsible for SLA decomposition and 
an SLA monitor that tracks relevant KPIs of each slice. If the SLA for a particular slice does not meet the minimum requirement, the monitor 
triggers the E2E orchestrator to perform a more feasible decomposition.
In 5G networks, \system's E2E orchestrator and SLA monitor can be implemented as Network Slice Management Functions 
while the domain orchestrators operate as Network Slice Subnet Management Functions (NSSMFs), 
as in Fig.~\ref{fig:design}. The NSMF coordinates, manages, and orchestrates E2E network slice instances across domains, 
while NSSMF handles lifecycle management of network slices within each domain. Fully compliant with the 3GPP architecture, \system
can be readily implemented in existing real-world 5G networks.

\begin{figure}[!t]
        \centering
\includegraphics[width=1.0\linewidth, height=1.8in]{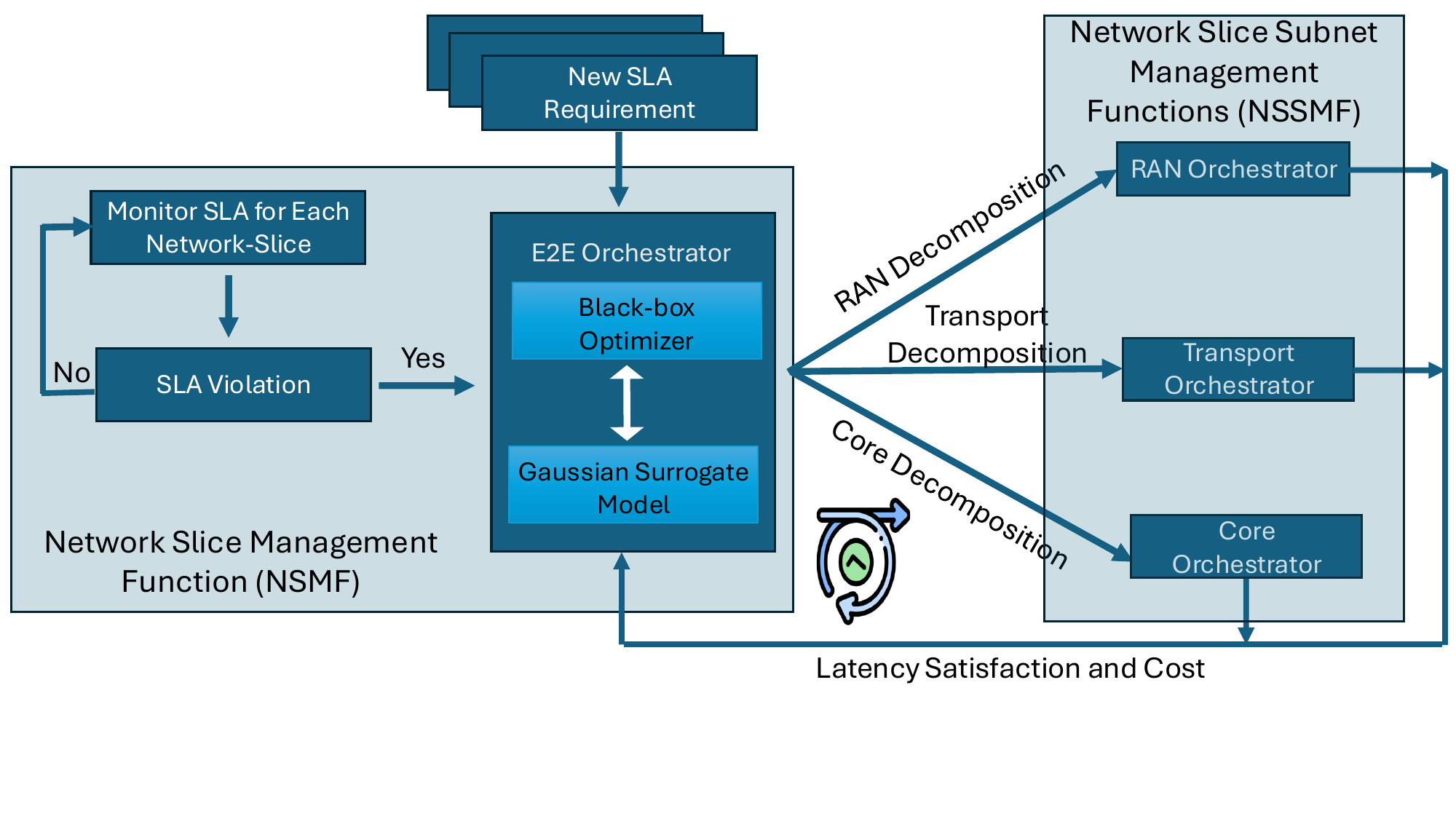}
        \caption{\system's 3GPP-compliant and UNEXT-native end-to-end system design
        }\label{fig:design}
\end{figure}

\subsection{Problem Formulation}
\label{sec:problem_form}
As described above, \system enables the E2E orchestrator, in coordination with domain orchestrators, to decompose an E2E SLA so as to minimize the total resource cost while meeting E2E SLA targets. For a better exposition, we explain \system's
algorithms with latency as the SLA of interest, while our formulation and learning algorithms in general are applicable to broader SLA metrics.

Let \( D \) represent the number of domains in an E2E network, and let \( L \) denote the E2E SLA latency target. The SLA decomposition is given by  
\( x := (x^1, \dots, x^D) \in (\mathbb{R}_+)^D \), where each \( x^i \) represents the latency allocated to domain \( i \) and $\mathbb{R}_+$ denotes all positive real numbers. 
The actual domain-level latency values are determined by functions \( f^1(x^1), \dots, f^D(x^D) \), while the corresponding costs 
associated with these latencies are given by \( g^1(x^1), \dots, g^D(x^D) \). 

The cost of meeting a decomposed latency target in a domain is determined by its orchestrator. 
We assume a robust, low-overhead channel between NSMF and NSSMF, allowing the E2E orchestrator to send latency targets to domains.
Upon receiving a latency target, each domain orchestrator implements its own resources-optimizations (unknown to  
the E2E orchestrator) while considering the stochastic conditions of the domain. It then reports the cost required to achieve 
the specified latency target back to the E2E orchestrator. If a domain cannot meet the given 
latency requirement, its orchestrator determines the best achievable latency along with the corresponding cost and send them to the E2E orchestrator.

The optimization problem (for the E2E orchestrator) is formulated as below in Eq.~\eqref{eqn:offline_opt}
in which the domain latency functions $f^1,\dots,f^D$ and cost functions $g^1,\dots,g^D$ are all unknown to the E2E orchestrator. 
The objective is to find an appropriate decomposition such that the E2E latency requirement is satisfied while the total cost (i.e., the resources consumed) is minimized. By definition, the E2E latency is the sum of all 
domain-level latencies, and the total cost is also the sum of all (normalized) domain-level costs. 

\begin{equation}\label{eqn:offline_opt}
\begin{aligned}
 \min_{x = (x^1,\dots,x^D)\in (\mathbb{R}_+)^D}~&\sum_{i=1}^D g^i(x^i),\\
\text{s.t.} ~&\sum_{i=1}^D f^i(x^i)\leq L.
\end{aligned}
\end{equation}

\begin{remark}
    The constraint in problem~\eqref{eqn:offline_opt} can be made more generic to incorporate multiple different SLA constraints of form 
   
    $$\mathrm{Comp}_j(f_j^1(x^1),\dots,f_j^D(x^D))\leq L_j,\forall j=1,\dots, N_{\mathrm{SLA}},$$
    where $\mathrm{Comp}_j(\cdot,\dots,\cdot)$ is the composition rule of the $j$-th SLA type, $L_j$ is the corresponding SLA target, and $N_{\mathrm{SLA}}$ is the total number of different SLAs. E.g., when the $j$-th SLA is E2E transmission failure probability, then $f_j^i(\cdot)$ denotes the transmission failure probability at domain $j$, and the composition rule becomes    $$\mathrm{Comp}_j(f_j^1(x^1),\dots,f_j^D(x^D))=1-\prod_{i=1}^D(1-f_j^i(x^i)).$$ 
\end{remark}

However, to optimize the above formulation, the E2E orchestrator must first learn the unknown functions representing latencies 
and associated costs across domains. Effectively integrating this learning process into the overall problem formulation is crucial for achieving an optimal solution.
Consequently, in the following section, we define an online optimization problem that captures the efficiency of the entire learning process.

Before introducing the detailed formulation, to make the online learning problem nontrivial and well-defined, we pose two common assumptions.
Since we are studying a constrained problem, it is necessary to assume that the problem is feasible.
\begin{assumption}
    Problem~\eqref{eqn:offline_opt} is feasible, i.e., there exists some $x= (x^1, \dots, x^D)\in (\mathbb{R}_+)^D$ such that $\sum_{i=1}^D f^i(x^i)\leq L$.
\end{assumption}
Moreover, in general, it is impossible to learn ``arbitrary'' functions efficiently. The assumption below restricts the functions of interest only to those ``well-behaved'' ones.
\begin{assumption}
It is assumed that $f^i\in \mathcal{H}_{k^{i,f}}, g^i\in \mathcal{H}_{k^{i,g}}$, where $k^{i,f}$ (respectively $k^{i,g}$) is some symmetric, positive semi-definite 
kernel function, and $\mathcal{H}_{k^{i,f}}$ (respectively $\mathcal{H}_{k^{i,G}}$) is the corresponding reproducing kernel Hilbert space (RKHS). Moreover, we assume $\lVert f^i \rVert_{k^{i,f}}\leq B^{i,f}$ and $\lVert g^i \rVert_{k^{i,g}}\leq B^{i,g}$, where $\lVert \cdot \rVert_{k^{i,f}}$ and $\lVert \cdot \rVert_{k^{i,g}}$ are the respective norms induced 
by the inner product in the corresponding RKHS.
\end{assumption}

\begin{remark}
    \system's design is specific to Assumption~3.1 holding true. In reality, if a slice is found infeasible prior to provisioning, the operator simply rejects it. In  case Assumption~3.2 fails, \system will suffer \emph{inevitable} additional regret and constraint violation because functions are represented in some RKHS up to some approximation error.
\end{remark}

\subsection{Formulation for Online Learning}
Given \system's interactive design, the action
for the E2E orchestrator in
each round $t$ is to determine an SLA decomposition, i.e., domain latency targets $x_t:=(x^1_t,\dots,x^D_t)$. Each domain then uses its internal resource optimizations, accounting for stochastic network conditions, to meet its assigned target.
The feedback received by the E2E orchestrator from each domain $i=1,\dots,D$  is given by
$y^{i,f}_t := f^i(x^i_t) + \epsilon^{i,f}_t$ and $y^{i,g}_t := g^i(x^i_t) + \epsilon^{i,g}_t$, where $\epsilon^{i,f}_t$ and $\epsilon^{i,g}_t$ 
are independent zero-mean sub-Gaussian noise with sub-Gaussian norms $R^{i,f}$ and $R^{i,g}$ respectively.

We define two metrics: regret $R_N$ and constraint violation $V_N$ to capture the efficiency of the learning process, defined as
\begin{align}
    R_N:=\sum_{t=1}^N \sum_{i=1}^D \left( g^i(x^i_t) - g^i(x^i_*)\right),
    \text{and}~V_N&:=\sum_{t=1}^N \left(\sum_{i=1}^D f^i(x^i_t) - L\right)^+,\label{eqn:con_vio_def}
\end{align}
where $x_* := (x^1_*,\dots,x^D_*)$ is an optimal solution of problem~\eqref{eqn:offline_opt}, and $x^+:=\max\{x,0\},\forall x\in\mathbb{R}$. Regret $R_N$ 
quantifies the cumulative excessive cost consumption, and $V_N$ reflects the cumulative (hard) SLA constraint violation over the course of $N$ interactions between the E2E orchestrator and domain orchestrators. 
An ideal learning algorithm for this online setting ensures both a low regret and a low constraint violation, implying that 
it is highly likely to converge on the optimal offline solution in a few iterations.

The black-box optimizer guides the E2E orchestrator to select latency decompositions, i.e., sending them to domain orchestrators to learn the unknown function, while maintaining a balance between \emph{exploration} and \emph{exploitation}: Since the learning algorithm starts without prior knowledge of underlying functions, it must explore by sampling from some regions to reduce uncertainty (exploration), even if it means occasionally
foregoing the best possible decision based on the current knowledge (exploitation). An effective algorithm should efficiently strike this balance to avoid suboptimal local minima and attain optimal solutions.

This essentially means that a nontrivial algorithm should ensure that both regret $R_N$ and constraint violation $V_N$ are sub-linear in the number of 
interactions $N$ (i.e., $R_N = o(N)$ and $V_N = o(N)$), meaning that the algorithm converges to an optimal solution, since the average 
regret and constraint violation both converge to zero asymptotically, i.e., $\lim_{N\rightarrow\infty}\frac{R_N}{N}\rightarrow 0$ and $\lim_{N\rightarrow\infty}\frac{V_N}{N}\rightarrow 0$.

{\textbf{Sampling heterogeneity:}}
\system's black-box optimizer uses distinct querying epochs for each domain, reflecting their different performance volatilities. In practice, RAN exhibits the highest volatility, 
requiring more frequent queries, while the core remains the most stable, necessitating fewer updates.
Motivated by such needs, \system naturally allows the E2E orchestrator to query each domain with different epoch lengths.
Assuming that the querying epoch length 
for a domain $i$ is given by $\tau^i$ (unit of time), the E2E orchestrator does not propose a new decomposition until
receiving at least one new sample from every domain - i.e., if $\tau^{\text{DM}}$ is the minimum time interval between two new decompositions proposed by the
E2E orchestrator, then $\tau^{\text{DM}}\geq \max_i \tau^i$.
Given a time window $T$, the regret and the constraint violation within it are defined as:

\begin{align}
    R(T)&:= \tau^{\text{DM}} \cdot R_{T/\tau^{\text{DM}}} = \tau^{\text{DM}}\cdot \sum_{t=1}^{T/\tau^{\text{DM}}} \left(\sum_{i=1}^D g^i(x^i_t) - g^i(x^i_*)\right),\\
    V(T)&:= \tau^{\text{DM}} \cdot V_{T/\tau^{\text{DM}}} = \tau^{\text{DM}}\cdot \sum_{t=1}^{T/\tau^{\text{DM}}} \left(\sum_{i=1}^D f^i(x^i_t) - f^i(x^i_*)\right)^+.
\end{align}

To see why these definitions make sense, essentially, they are equal to 
the standard $T/\tau^{\text{DM}}$-round regret and constraint violation, \emph{scaled by} the length of the time interval ($\text{DM}$) to reflect that now 
the ``weight'' of each decision is (proportional to) $\tau^{\text{DM}}$.
Roughly speaking, our regret/constraint violation definition takes the effect of decision-making epoch into consideration and precisely reflects the performance within the time window $T$.

\section{Gaussian Process Surrogate Model}
\label{sec:gp_surrogate}
To model the unknown functions ($\{f^i\}_{i=1}^D$ and $\{g^i\}_{i=1}^D$), \system employs Gaussian Processes (GPs) as surrogate models. GPs not only predict unknown 
functions, but also quantify uncertainty in these predictions, allowing BO algorithms to effectively balance between exploration and exploitation. In addition, GPs offer computational advantages through closed-form updates.

\subsection{Learning Unknown Function $f^i$}
Conditioned on the set of all history observations $H^i_t := \{x^i_s,y^{i,f}_s\}_{s=1}^t$, the posterior distribution for $f^i$ is $\mathcal{GP}(\mu^{i,f}_t(\cdot),k^{i,f}_t(\cdot,\cdot))$, where
\begin{align}
    &\mu^{i,f}_t(x):=k^{i,f}_{t}(x)^\top(K^{i,f}_t+\lambda I_t)^{-1}Y^{i,f}_t, \label{eqn:GP_f_mean}\\
    &k^{i,f}_t(x,x'):=k^{i,f}(x,x') - k^{i,f}_t(x)^\top (K^{i,f}_t+\lambda I_t)^{-1} k^{i,f}_t(x'),\label{eqn:GP_f_var}
\end{align}
in which
$k^{i,f}_t(x):=[k^{i,f}(x^i_1,x),\dots,k^{i,f}(x^i_t,x)]^\top$, $k^{i,f}$ is the kernel for function $f^i$, $K^{i,f}_t := [k(x_m^i, x_n^i)]_{m,n=1,\dots,t}$, $I_t$ is the $t$-dim. identity matrix, and $Y^{i,f}_t$ is the noisy observation vector $[y^{i,f}_1,\dots,y^{i,f}_t]^\top$. We also define $\sigma^{i,f}_t(x)^2 := k^{i,f}_t(x,x)$. Under the concrete context of latency, vector $Y^{i,f}_t$ contains all (noisy) latency observations up to the $t$-th interaction from domain $i$, $\mu^{i,f}_t(x)$ is the estimated latency if we specify $x$ as a domain-level SLA to domain $i$, and $\sigma^{i,f}_t(x)^2$ represents the uncertainty at $x$. We let $K^{i,f}(A):=[k^{i,f}(x,x')]_{x,x'\in A}$ for any finite set $A\subset\mathcal{X}$, then we use $\gamma_t(k^{i,f}, \mathcal{X})$ to denote the \emph{information gain}, which is defined as $\gamma_t(k^{i,f}, \mathcal{X}):=\max_{A\subseteq\mathcal{X}:|A|=t} \frac{1}{2}\ln|I_t + \lambda^{-1}K^{i,f}(A)|$ and plays an important role in the theoretical analysis. It depends on both the kernel $k^{i,f}$ and domain $\mathcal{X}$ (which is $(\mathbb{R}_+)^D$ in this work), and is simplified as $\gamma_t$ whenever the context is clear. 
\begin{remark}
   In this work, \system assumes all the unknown functions to have the same kernel. However, \system's design framework can readily handle distinct kernels across different unknown functions. 
\end{remark}

\subsection{Learning Unknown Function $g^i$}
In order to learn the unknown function $g^i$, \system uses another GP surrogate model calculated similarly to the way discussed above, based on the set of observations $\widetilde{H}^i_t = \{x^i_s,y^{i,g}_s\}_{s=1}^t$, i.e.,

\begin{align}
    &\mu^{i,g}_t(x):=k^{i,g}_{t}(x)^\top(K^{i,g}_t+\lambda I_t)^{-1}Y^{i,g}_t,\label{eqn:GP_g_mean}\\
    &k^{i,g}_t(x,x'):=k^{i,g}(x,x') - k^{i,g}_t(x)^\top (K^{i,g}_t+\lambda I_t)^{-1} k^{i,g}_t(x').\label{eqn:GP_g_var}
\end{align}
\begin{remark}
   When all functions share the same kernel, all posterior variances $\{k^{i,f}_t(\cdot,\cdot),k^{i,g}_t(\cdot,\cdot)\}_{i=1}^D$ are identical.
\end{remark}

\section{Algorithm Design and Theoretical Analysis}
\subsection{Algorithm Overview}
\system's algorithm leverages the framework for constrained BO in \cite{pmlr-v202-xu23h}. The main idea is to keep solving the optimization problem~\eqref{eqn:offline_opt}, 
while replacing the unknown true functions with the optimistic/pessimistic estimates built using the GP surrogate models based on (noisy) observations. 
As more feedback is collected from the domain orchestrators over time, the estimates are improved
with the decisions converging to the optimal one.

The complete pseudocode is given in Algorithm~\ref{alg:BO}. 
\begin{algorithm}[t]
    \caption{\system: Bayesian Optimization  for SLA decomposition}\label{alg:BO}
    \begin{algorithmic}[1]
    \STATE \textbf{Input:} E2E SLA $L>0$, kernel function $k$, the number of domains $D$, decision-making epoch length $\tau^{\text{DM}}$, querying epoch length $\tau^{1},\dots,\tau^{D}\leq\tau^{\text{DM}}$, noise sub-Gaussian norms $\{R^{i,f},R^{i,c}\}_{i=1}^D$, time horizon $T$, failure probability $\delta\in(0,1)$
       \STATE \textbf{Define: } the number of interactions $N=T/\tau^{\text{DM}}$, $\beta^{i,f}_t := B^{i,f} + \frac{R^{i,f}}{\sqrt{\text{DM}/\tau^i}}\sqrt{2\left(\gamma_t+1+\ln(2D/\delta)\right)}$, and $\beta^{i,g}_t := B^{i,g} + \frac{R^{i,g}}{\sqrt{\text{DM}/\tau^i}}\sqrt{2\left(\gamma_t+1+\ln(2D/\delta)\right)}$
    \FOR{$t =1,\dots,N$}
     \STATE From each domain $i$, collect $\tau^{\text{DM}}/\tau^i$ independent noisy samples evaluated at $x^i_t$, and calculate the averages $(\overline{y}^{i,f}_t, \overline{y}^{i,g}_t)$ 
     \STATE Update surrogate models as in Eqs.~\eqref{eqn:GP_f_mean}, \eqref{eqn:GP_f_var}, \eqref{eqn:GP_g_mean}, and \eqref{eqn:GP_g_var}
     \STATE Determined new decomposed SLAs $x_t=(x^1_t,\dots,x^D_t)\in \text{argmin}_x \sum_{i=1}^D (\mu_t^{i,g}(x^i) -\beta_t^{i,g}\cdot \sigma_t^{i,g}(x^i))$ subject to $\sum_{i=1}^D (\mu_t^{i,f}(x^i) -\beta_t^{i,f}\cdot \sigma_t^{i,f}(x^i))-L\leq 0$, and apply them to all domains
    \ENDFOR
    \end{algorithmic}
\end{algorithm}
Given the horizon $T$ and decision-making interval $\tau^{\text{DM}}$, the algorithm performs $T/\tau^{\text{DM}}$ rounds of interactions. Before the $t$-th decision is made, there are $\tau^{\text{DM}}/\tau^i$ newly collected samples from domain $i$, all of which are (independent) noisy feedback evaluated at $x^i_{t-1}$. The samples are first averaged and then used for updating the surrogate model. The averaging trick simplifies the theoretical analysis, and saves
computational cost \cite{li2023private}.

\textbf{Computational complexity:}
The surrogate model update can be done efficiently via the ``recursive update''~\cite{chowdhury2017kernelized}. To determine the next decomposition (i.e., solving the auxiliary constrained problem), if the (search) space is discrete, one can simply iterate over it. If the space is continuous, one can also use efficient common practice in the literature~\cite{pmlr-v202-xu23h}: When the dimension is low, one can discretize the space with uniform grid; when the dimension is high, one can apply gradient-based methods or multiple re-starts with random starting points~\cite{wilson2018maximizing}. For each slice, GP update inherently incurs $N$-step total computational complexity of order $O(N^3)$. If Odin is deployed with admission control prior to slice provisioning, operators can potentially limit the total number of slices.

\textbf{Communicational efficiency:}
For each slice, \system requires only $2D$ scalars transmitted between domain orchestrators and the E2E orchestrator, resulting in negligible overhead over high-capacity optical networks, even if there are a large number of slices.

\subsection{Theoretical Analysis}

We first state the concentration result, on which the final regret analysis relies heavily. It roughly says how close the surrogate model is to the true unknown function, or at most how large the estimation error could be at any time $t$. Intuitively, the estimation error should decrease over time as more experience is gained, and lower estimation error should introduce lower regret and constraint violation. In the context of our problem, for any decomposed domain-level latency, it upper bounds the difference between the actual latency (cost resp.) determined by the surrogate model and that determined by the true function.
\begin{lemma}[Concentrations]\label{lemma:conc_ineq}
    With probability at least $1-\delta$, it holds for any $t\leq N$, decision $x$, and domain $i$ that 
    $|f^{i}(x)-\mu_t^{i,f}(x)|\leq \beta^{i,f}_t \cdot \sigma_t^{i,f}(x)$ and $|g^{i}(x)-\mu_t^{i,g}(x)|\leq \beta^{i,g}_t \cdot \sigma_t^{i,g}(x)$,
    where
    \begin{align*}
        \beta^{i,f}_t &:= B^{i,f} + \frac{R^{i,f}}{\sqrt{\tau^\text{DM}/\tau^i}}\sqrt{2\left(\gamma_t+1+\ln(2D/\delta)\right)},\\
        \beta^{i,g}_t& := B^{i,g} + \frac{R^{i,g}}{\sqrt{\tau^\text{DM}/\tau^i}}\sqrt{2\left(\gamma_t+1+\ln(2D/\delta)\right)}.
    \end{align*}
\end{lemma}
\begin{proof}[Proof of Lemma~\ref{lemma:conc_ineq}]
    For any sub-Gaussian norm $R'>0$ and number of samples $n\geq1$, the average of $n$ independent $R'$-sub-Gaussian samples is $\frac{R'}{\sqrt{n}}$-sub-Gaussian. The claimed results are then implied by Lemma~2.4 in \cite{pmlr-v202-xu23h} via replacing the sub-Gaussian norm.
\end{proof}

\noindent In the lemma above, we show upper bounds on the deviation at time $t$. Another key lemma we need is also standard for BO, which bounds cumulative deviation over time and is formally stated below. The key role of this lemma is that, eventually our regret and constraint violation are controlled by such cumulative deviations.
\begin{lemma}[Lemma~4 in \cite{chowdhury2017kernelized}]
    For any domain $i$, the sequence $x^i_1,\dots,x^i_N$ chosen by Algorithm~\ref{alg:BO} satisfies both
    $$
        \sum_{t=1}^N \sigma^{i,f}_t(x^i_t)\leq 2\sqrt{(N+2)\gamma_{N}} \text{~and~}
         \sum_{t=1}^N \sigma^{i,g}_t(x^i_t)\leq 2\sqrt{(N+2)\gamma_{N}}.
    $$
\end{lemma}

Finally, we are able to show our main results, i.e, the respective upper bounds on regret and constraint violation, formally below.
\begin{theorem}[Main Theoretical Guarantee]\label{thm:reg_and_vio_bound}
With probability at least $1-\delta$, Algorithm~\ref{alg:BO} ensures 
$$R(T) = O\left(\tau^{\text{DM}}\sqrt{T\cdot \gamma_{T/\tau^{\text{DM}}}}\cdot \sum_{i=1}^D\left(B^{i,g} + \frac{R^{i,g}}{\sqrt{\tau^{\text{DM}}/\tau^i}}\sqrt{\gamma_{T/\tau^{\text{DM}}}}\right)\right)$$ $$V(T) =O\left(\tau^{\text{DM}}\sqrt{T\cdot \gamma_{T/\tau^{\text{DM}}}}\cdot \sum_{i=1}^D\left(B^{i,f} + \frac{R^{i,f}}{\sqrt{\tau^{\text{DM}}/\tau^i}}\sqrt{\gamma_{T/\tau^{\text{DM}}}}\right)\right).$$
\end{theorem}

\begin{remark}
    The bounds stated above capture the impact of both the decision-making epoch length $\tau^{\mathrm{DM}}$ and the domain heterogeneity in terms of noise level ($R^{i,f},R^{i,g}$), smoothness of functions' behavior ($B^{i,f},B^{i,g}$), and data querying epoch ($\tau^i$) on the algorithm's performance. This offers an informative theoretical toolbox on what performance to expect given all these conditions, or how frequently each domain should be queried (to ``compensate'' the negative impact from high noise level and/or non-smooth behavior).
\end{remark}
\begin{remark}
    The bounds are indeed sub-linear for common kernels. E.g., when all functions are linear (i.e., $k$ is the linear kernel function), we have $\tau_t=O(\ln t)$, and the bounds become $$R(T) = O\left(\sqrt{T\ln T}\cdot \sum_{i=1}^D\left(B^{i,g} + R^{i,g}\sqrt{\ln T}\right)\right)$$ and $V(T) = O\left(\sqrt{T\ln T}\cdot \sum_{i=1}^D\left(B^{i,f} + R^{i,f}\sqrt{\ln T}\right)\right),$
    when $\tau^{\text{DM}}=\tau^1=\dots=\tau^D=1$. Both bounds are $\widetilde{O}(D\sqrt{T})$ and hence sub-linear, implying a provable convergence to the optimal solution.
\end{remark}

\begin{proof}[Proof of Theorem~\ref{thm:reg_and_vio_bound}]
   This proof is an adaptation from that of Theorem~4.3 in \cite{pmlr-v202-xu23h}. We first bound regret (and constraint violation) in terms of the number of interactions $N=T/\tau^{\text{DM}}$.

   Let $r_t:=\sum_{i=1}^D \left( g^i(x^i_t) - g^i(x^i_*)\right)$ be the instantaneous regret incurred in the $t$-th decision. It can be upper bounded as
   \begin{align}
       r_t &= \sum_{i=1}^D \left( g^i(x^i_t) - g^i(x^i_*)\right)\nonumber\\
       &\overset{\text{(a)}}{\leq} \sum_{i=1}^D \left( \mu^{i,g}(x^i_t) + \beta_t^{i,g}\cdot \sigma_t^{i,g}(x_t^i) \right)-\sum_{i=1}^D \left( \mu^{i,g}(x^i_*) - \beta_t^{i,g}\cdot \sigma_t^{i,g}(x_*^i)\right)\nonumber\\
       &\overset{\text{(b)}}{\leq} \sum_{i=1}^D \left( \mu^{i,g}(x^i_t) + \beta_t^{i,g}\cdot \sigma_t^{i,g}(x^i_t) \right)-\sum_{i=1}^D \left( \mu^{i,g}(x^i_t) - \beta_t^{i,g}\cdot \sigma_t^{i,g}(x^i_t)\right)\nonumber\\
       &= 2\cdot \sum_{i=1}^D \beta_t^{i,g}\cdot \sigma_t^{i,g}(x^i_t),
   \end{align}
where step (a) is due to Lemma~\ref{lemma:conc_ineq} and step (b) holds since
\begin{align}
    \sum_{i=1}^D (\mu_t^{i,f}(x_*^i) -\beta_t^{i,f}\cdot \sigma_t^{i,f}(x_*^i)) \leq \sum_{i=1}^D f^i(x_*^i)\leq L.
\end{align}
   
    Now we can readily bound the cumulative regret as
    \begin{align}
        R_N = \sum_{t=1}^N r_t
        \leq 2\cdot \sum_{t=1}^N\sum_{i=1}^D \beta_t^{i,g}\cdot \sigma_t^{i,g}(x^i_t)&\leq 2\cdot \sum_{i=1}^D\beta_N^{i,g}\sum_{t=1}^N  \sigma_t^{i,g}(x^i_t)\nonumber\\
        &\leq 4\cdot \sum_{i=1}^D\beta_N^{i,g}\sqrt{(N+2)\gamma_{N}}.\nonumber
    \end{align}
    Finally, we arrive at the bound on $R(T)$ by replacing $N$ and $\beta_N^{i,g}$ with their own definitions respectively, and multiplying $R_N$ by $\tau^{\text{DM}}$.

    Similarly, we define $v_t:= \left( \sum_{i=1}^D f^i(x^i_t) - L\right)^+$ as the instantaneous (hard) constraint violation incurred in the $t$-th decision. It can be bounded as
    \begin{align}
        v_t & \overset{\text{(a)}}{\leq} \left( \sum_{i=1}^D \left( \mu_t^{i,f}(x^i_t) + \beta_t^{i,f}\cdot \sigma_t^{i,f}(x^i_t)\right) - L\right)^+\nonumber\\
        &\leq \left( \sum_{i=1}^D \left(\mu_t^{i,f}(x^i_t) - \beta_t^{i,f}\cdot \sigma_t^{i,f}(x^i_t)\right) - L\right)^+ + 2\beta_t^{i,f}\cdot \sigma_t^{i,f}(x^i_t)\nonumber\\
        &\overset{\text{(b)}}{\leq} 2\beta_t^{i,f}\cdot \sigma_t^{i,f}(x^i_t),
    \end{align}
where step (a) is from Lemma~\ref{lemma:conc_ineq} and step (b) holds simply due to the algorithm design:
\begin{align}
     \sum_{i=1}^D \left(\mu_t^{i,f}(x^i_t) - \beta_t^{i,f}\cdot \sigma_t^{i,f}(x^i_t)\right) - L\leq 0.
\end{align}

\begin{figure*}[t]
  \centering
      \includegraphics[width=0.8\textwidth, height=1.8in]{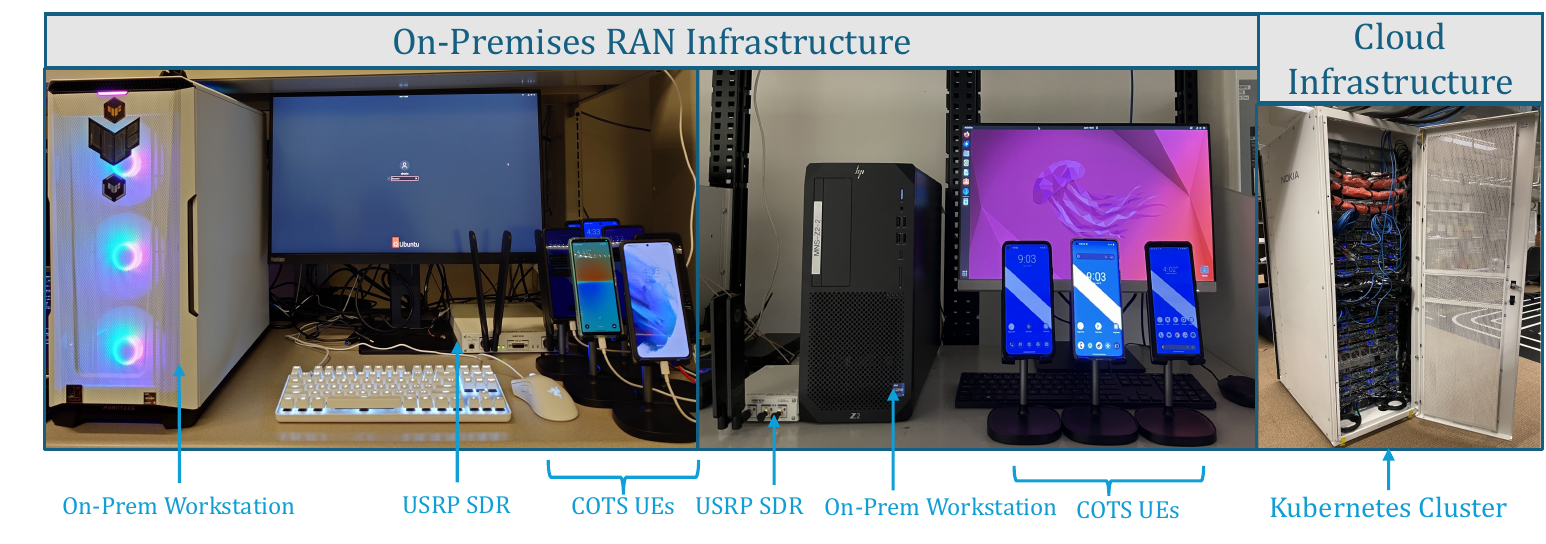}
  \caption{RAN and core testbed setup used in our trace data collection}
  \label{fig:testbed}
\end{figure*}
    
Lastly, we can bound cumulative constraint violation $V(T)$ via same path from $r_t$ to $R(T)$.
\end{proof}
Sub-linear bounds on regret and constraint violation mean that, the time-averaged performance gap between \system's algorithm and the optimal solution is diminishing over time. Asymptotically, this gap goes to zero, meaning that \system's algorithm performs on average as well as the optimal decomposition.

\section{System Implementation}
\label{sec:implement}
We conduct both trace-based evaluation as well as simulations to capture the efficacy of \system, comparing its performance to baseline solutions. 
The trace-based evaluations use data collected from our RAN, core, and transport testbeds that mimic real-world deployments. Below, we discuss our testbed implementation for domain-data collection and the procedure for trace data collection. 

\noindent\textbf{RAN testbed:}
We build an over-the-air experimental Open RAN (O-RAN) testbed using 5 HexRAN~\cite{kak2024hexranprogrammableapproachopen} 
gNodeBs (gNBs), the HexRIC~\cite{10.1145/3572864.3580329} RAN controller, and 20 commercial smartphones 
evenly distributed across the five gNBs (Fig.~\ref{fig:testbed}).
Each gNB runs on a 
ThinkStation P5 workstation augmented with a USRP X310 SDR which broadcasts a 100MHz n78 carrier capable of 500Mbps 
downlink throughput. The HexRIC controller is used both to vary the number of availabel PRBs 
 allocated to users and to collect trace collection downlink RAN-level latency traces. Within this context, we execute five distinct scenarios by varying the downlink traffic for each user 
from 25Mbps to 125Mbps in 25 Mbps increments, resulting in aggregate traffic from 500 Mbps to 2500 Mbps. In each scenario, we fix a different MCS via HexRIC and dynamically vary the PRB count, measuring latency for each MCS–PRB combination. 

\noindent\textbf{Transport network testbed:}
We simulate the transport domain using Mininet~\cite{mininet}, modeling it as a fully connected graph with six hosts and eight switches. This topology ensures multiple paths exist between any two hosts. To replicate real-world network conditions, we introduce background UDP traffic between all host pairs at distinct rates using the MGEN traffic generator~\cite{mgen}.

The primary latency measurement is conducted between two designated hosts — one representing an interface to the core's UPF and the other to the RAN — using \emph{Two-Way Active Measurement Protocol} (TWAMP)~\cite{hedayat2008two}. The cost metric is defined as the minimum path bandwidth required to meet a specified latency threshold, where the path bandwidth is determined by the lowest link-bandwidth (measured using \emph{Bandwidth Monitor NG} \cite{bmng}) along a given path. The transport orchestrator's role is to select the least expensive path that best meets the latency requirement.

\noindent\textbf{Core network testbed:}
The core network testbed (see Fig.~\ref{fig:testbed}) is built using network functions from the Open5GS project~\cite{open5gs}. We deploy the control plane functions (i.e., AMF, SMF, etc.) on a single HPE ProLiant DL360 Gen11 server, while the User Plane Function (UPF) is deployed on a workstation with an Intel i9-13900KS CPU capable of a 6~GHz maximum clock speed. We use \emph{cpupower}~\cite{Canonical} to vary the CPU frequency of UPF workstations against the backdrop of varying user plane traffic. We then measure the corresponding core-level user plane latency for each combination between CPU frequency and traffic value.

\noindent\textbf{Trace data collection:}
We collect data from each domain independently and combine them offline to mimic the core-to-RAN dataflow. In the RAN testbed, downlink latency from gNB to UE is measured for fixed MCS and varying PRBs using Poisson UDP traffic from MGEN~\cite{mgen}. We use MGEN~\cite{mgen} for UDP traffic generation that follows a Poisson distribution. 
For the core dataset, we collect the latency values for the UPF, which handles the user data packets. We measure the latency values between
the N3 and N6 interfaces of the core. We vary the CPU frequency from 1~GHz to 5.5~GHz and record the CPU usage as it handles processing the data traffic. The transport domain latency
is collected between pairs of hosts for different combinations of paths and bandwidths.

The data collected from the three domains are used to model the unknown
(black-box) performance functions ($\{f^i\}_{i=1}^D$) while the black-box cost functions ($\{g^i\}_{i=1}^D$) for each domain are modeled using the number of PRBs, link bandwidth, and the CPU 
frequency data for the RAN, transport and the core domains, respectively. 
\section{Evaluations}
\label{sec:eval}
In this section, we present \system's performance in trace-based evaluations followed by large-scale simulations.
In our evaluations, there is no need to choose the kernel function or tail parameters, since the kernel (matrix) is determined \emph{empirically}, following \cite{chowdhury2017kernelized}.
Also, our focus is on each domain’s intrinsic heterogeneity (e.g., latency/cost behavior), therefore, we ignore query epochs by setting $\tau^{\text{DM}}=\tau^1=\dots=\tau^D=1$. We compare \system with several baseline solutions, which are introduced below.

\begin{figure}[!t]   
  \centering           
      \includegraphics[width=0.50\textwidth, height=1.95in]{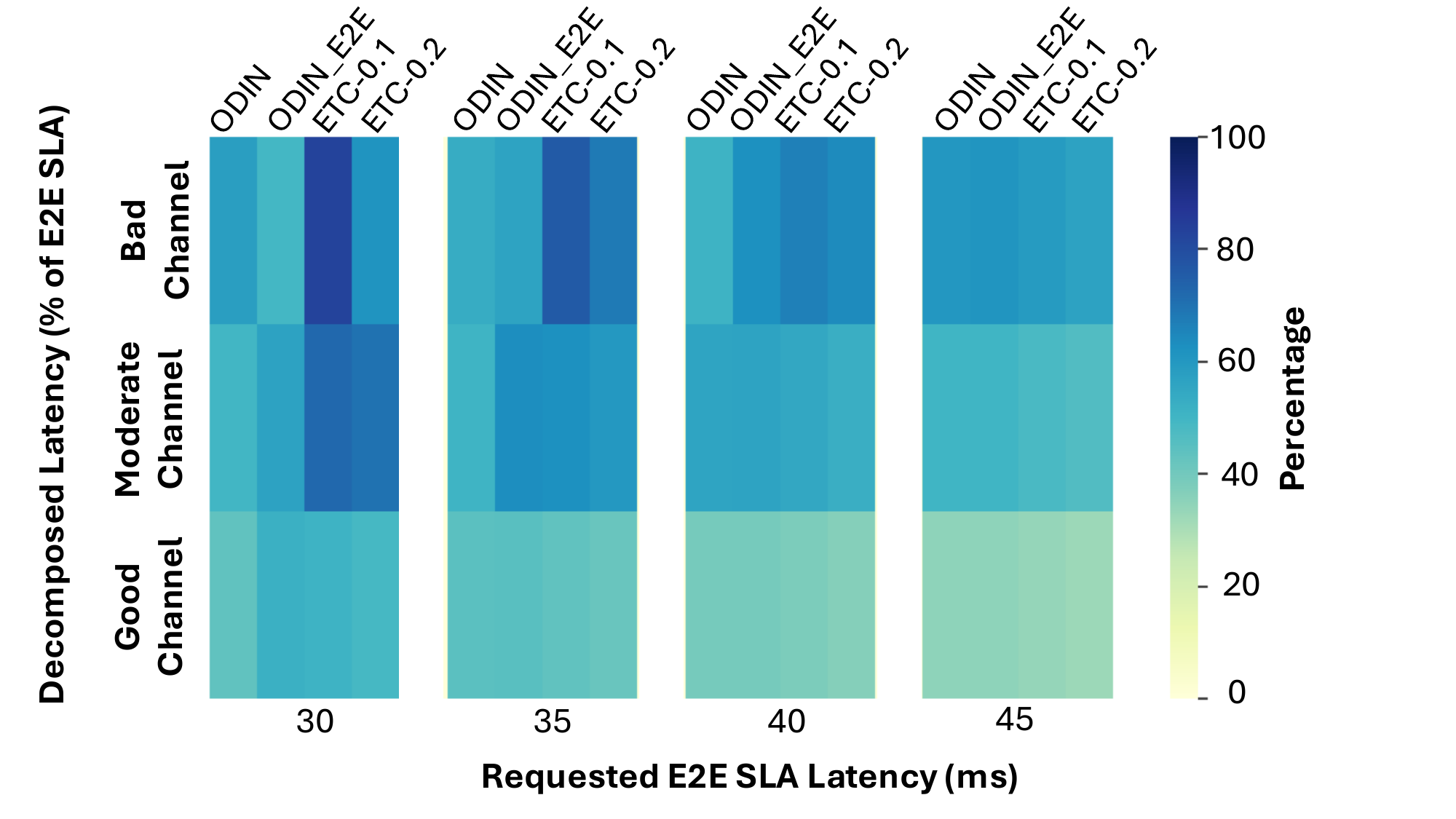}
  \caption{Latency decomposition in the RAN under various channel conditions} 
  \label{fig:heatmap_ran}           
\end{figure}

\begin{figure*}[!t]   
  \subfloat[Regret] 
  {
      \label{fig:trace_reg}\includegraphics[width=0.33\textwidth]{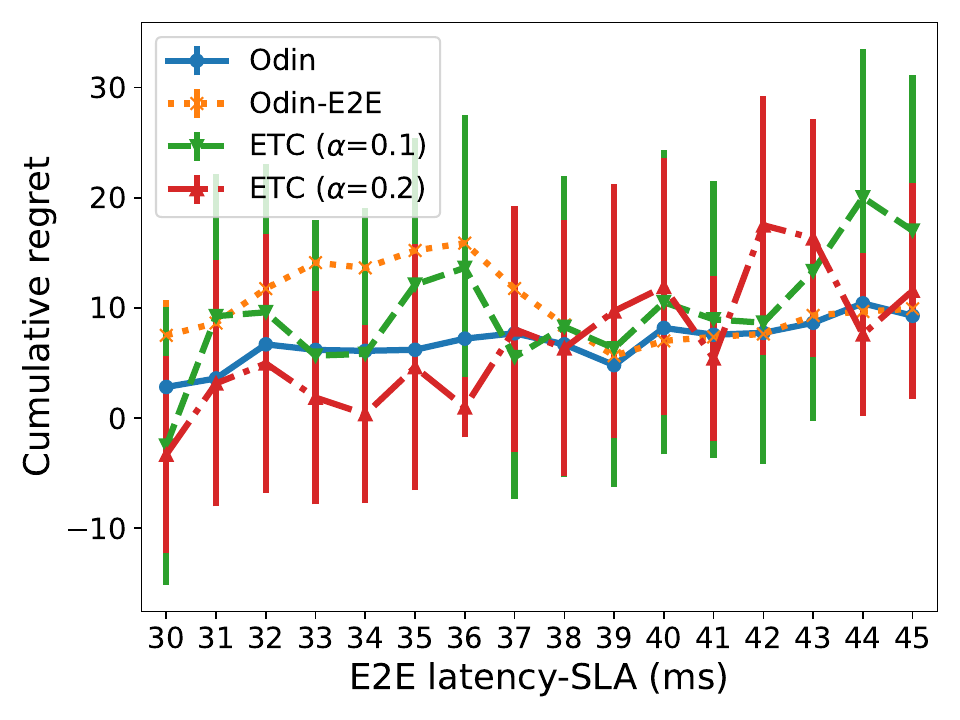}
  }
  \subfloat[Constraint violation]
  {
      \label{fig:trace_con_vio}\includegraphics[width=0.33\textwidth]{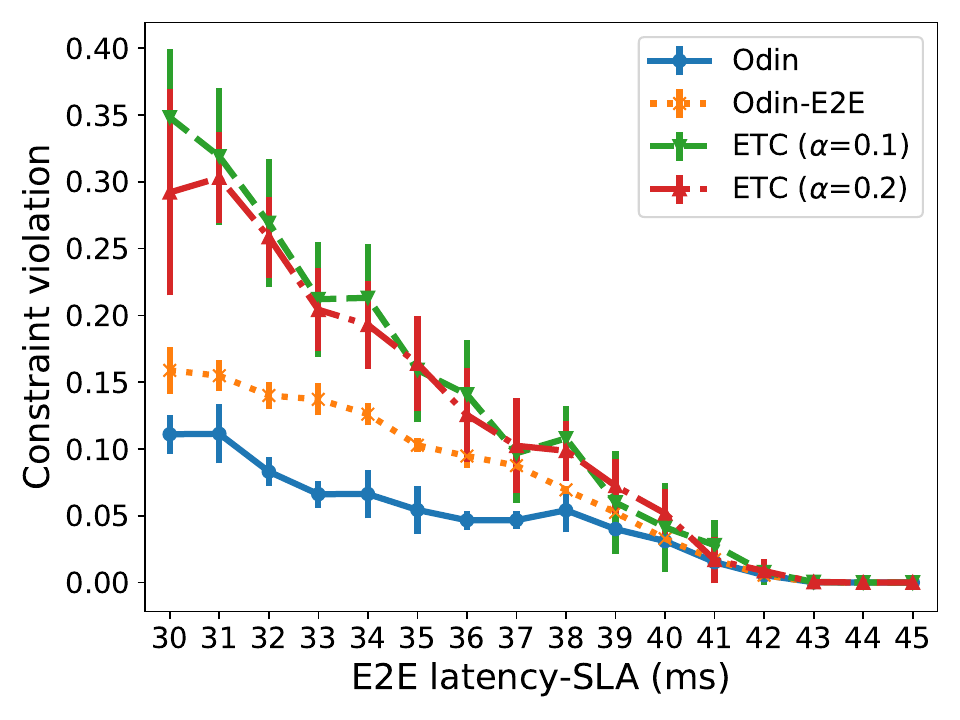}
  }
  \subfloat[E2E latency performance] 
  {
      \label{fig:trace_lat}\includegraphics[width=0.33\textwidth]{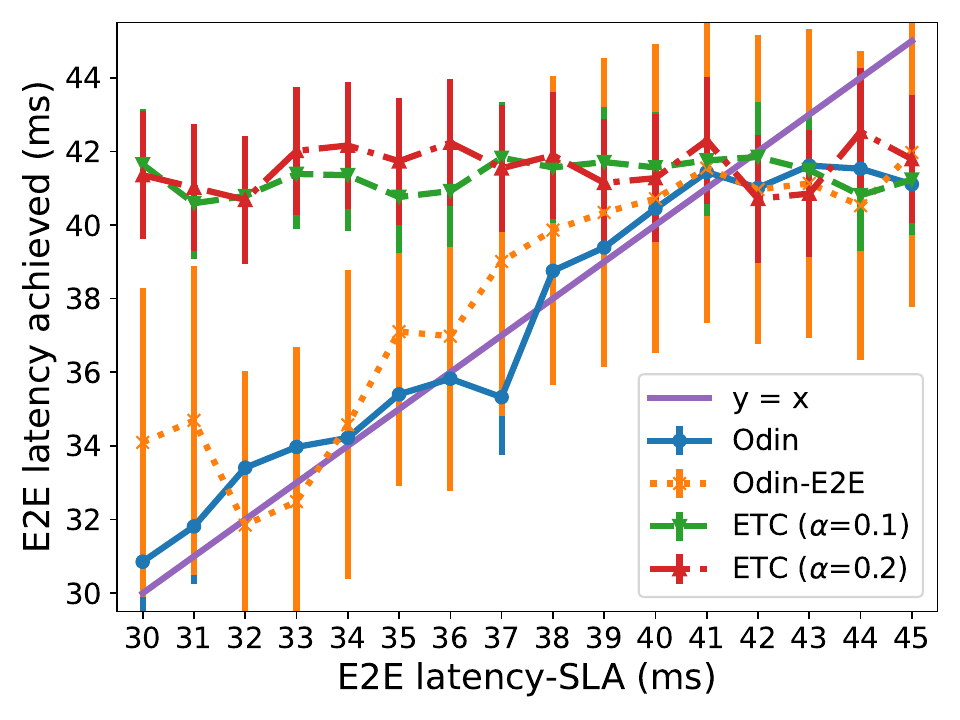}
  }

\caption{\system vs. baselines in trace-based evaluation} 
  \label{fig:trace_data}           
\end{figure*}
\subsection{Baseline Solutions}
\noindent {\bf {\system-E2E:}} This is the E2E variant of \system's design in Alg.~\ref{alg:BO}. The only difference being that \system-E2E treats the entire
network as a single ``bigger'' black-box function - i.e., in online interactions, it observes the E2E latency and total costs combined from all domains, instead of estimating them per-domain.
As a BO problem, it builds GP surrogate models directly on the unknown (E2E) performance function $f':=\sum_{i=1}^D f^i$ and cost function $g':=\sum_{i=1}^D g^i$.

\noindent {\bf {Explore-then-commit (ETC):}} ETC is obtained by adapting the design in \cite{de2020decomposing,hsu2023sla} to the online case. For a given $\alpha\in(0,1)$, 
the ETC uniformly explore the search space, observing the outputs of the unknown functions for the first $\alpha T$ rounds. The ETC
uses these observations to build its surrogate models. In the remaining $(1-\alpha)T$ rounds, the ETC uses these surrogate models as a replacement for the
unknown functions solving the optimization defined in Eq.~\eqref{eqn:offline_opt}.
In our evaluations, we use $\alpha = 0.1$  and $0.2$.
\subsection{Trace-Based Evaluation}
\label{sec:trace_eval}

We conduct our trace-based evaluations using the data collected on our testbeds. First, we show \system's ability to adapt the 
decomposition to the varying domain conditions. 

For this experiment, we use the RAN data with fixed MCS = 16 (bad SINR) to capture the bad channel effects on the performance and cost, while for the moderate
and the good channel conditions in RAN, we use the data collected with MCS = 22 and MCS = 27, respectively. We keep the transport and core conditions fixed.

The heatmap in Fig.~\ref{fig:heatmap_ran} shows each algorithm's latency decomposition for the RAN (as a percentage of the E2E target) under different channel 
conditions. While \system-E2E 
ignores the RAN channel condition and makes similar decompositions irrespective of RAN performance, \system adapts the decomposition to the RAN domain conditions. This can be best observed for E2E SLA targets of 30~ms and 35~ms. \system provides 
the RAN with a
latency target that is relatively higher when the channel is bad (darker shade, RAN's share of decomposed latency is 56\% of E2E latency when SLA = 30ms) but the latency share for RAN is reduced (i.e., lighter shade, 38\% of E2E latency when SLA = 30ms) when the channel is good. Similarly, latency share for the RAN when the channel is moderate is in-between the good and bad RAN conditions (e.g., 38\% of E2E latency when SLA = 30ms). While the ETC algorithms on the other hand decomposes providing the RAN a latency target higher than what \system provides when the channel is bad (78\% of E2E latency when SLA=30ms), but  as we observe in Fig.~\ref{fig:trace_lat}, ETC-0.1 achieves the E2E latency of 42~ms when the SLA is 30~ms mainly because the other two domains (core and the transport) fail to achieve the small latency targets given to them. \system performs the latency decomposition by taking the domain conditions whilst still ensuring that the SLA is not violated.

We next plot the (cumulative) regret\footnote{Regrets can be negative due to the formulation in problem~\eqref{eqn:offline_opt}: Some inputs could produce 
negative per-round regret by violating the constraint.} and  constraint violation\footnote{In all figures, we plot the cumulative constraint violation \emph{normalized} by the E2E SLA value, i.e., $V_T$ defined as in Eq.~\eqref{eqn:con_vio_def} further \emph{divided by} $LT$.} in the graphs shown in Fig.~\ref{fig:trace_data} for various E2E SLA targets ranging from
$\{30,\dots,45\}$ (which span the latency range observed in our dataset). We fix the number of rounds as $T=100$ and conduct 10 independent runs for each experiments, plotting mean and standard deviation.
These metrics reflect the efficiency of each learning process: regret measures cost difference from the optimal solution, while constraint violation captures SLA violation relative to optimal. While the baselines (e.g., ETC with $\alpha=0.2$) have a lower regret for various latency SLAs, their
corresponding constraint violation for those latencies are significantly higher. \system, on the other hand, has lower regret (indicating its cost usage is closer to optimum) while ensuring the lowest constraint violation. This consequently ensures that \system can find a decomposition such that the E2E latency 
for a given is satisfied (Fig.~\ref{fig:trace_lat}).
On the other hand, the baseline ETC algorithms, which rely on only $\alpha T$ samples to construct the surrogate model, is mostly flat (Fig.~\ref{fig:trace_lat}),
indicating its failure to accurately learn the unknown (black-box) functions.

\begin{figure*}[t]   
  \centering           
  \subfloat[Regret]
  { \label{fig:syn_reg_low_short}\includegraphics[width=0.25\textwidth]{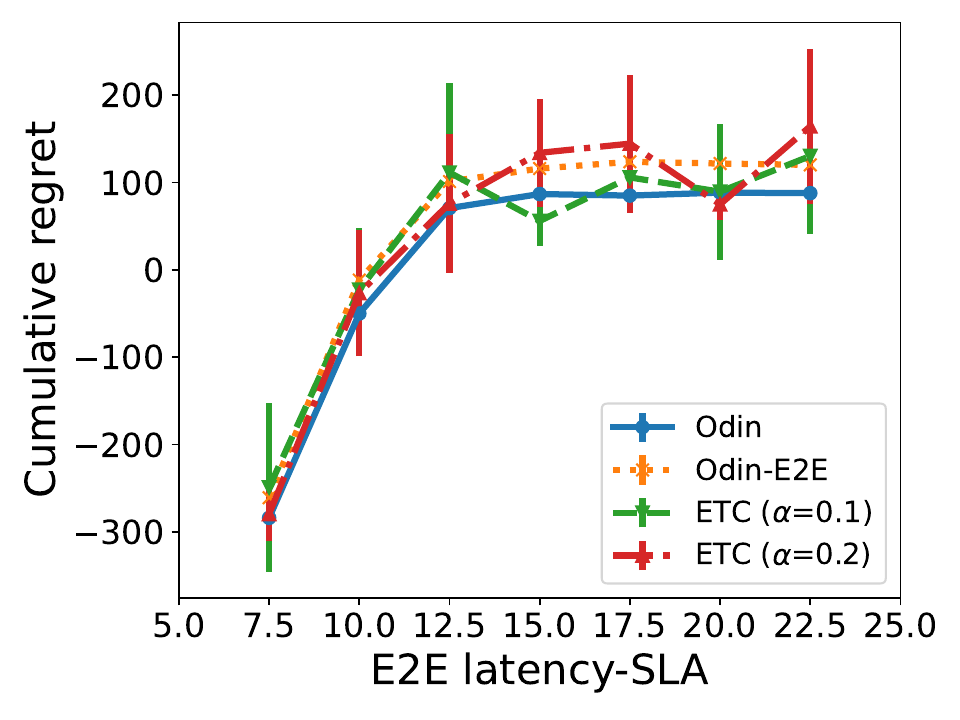}}
  \subfloat[Constraint violation]
{\label{fig:syn_con_vio_low_short}\includegraphics[width=0.25\textwidth]{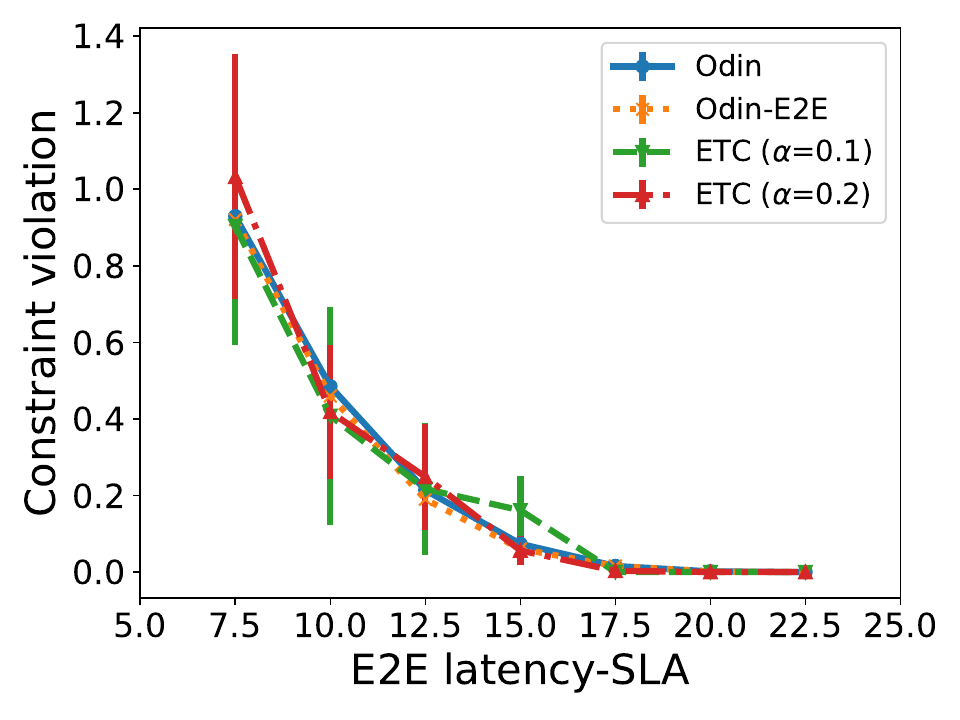}}
  \subfloat[No. (\%) rounds with zero violation] 
  {\label{fig:syn_satis_low_short}\includegraphics[width=0.25\textwidth]{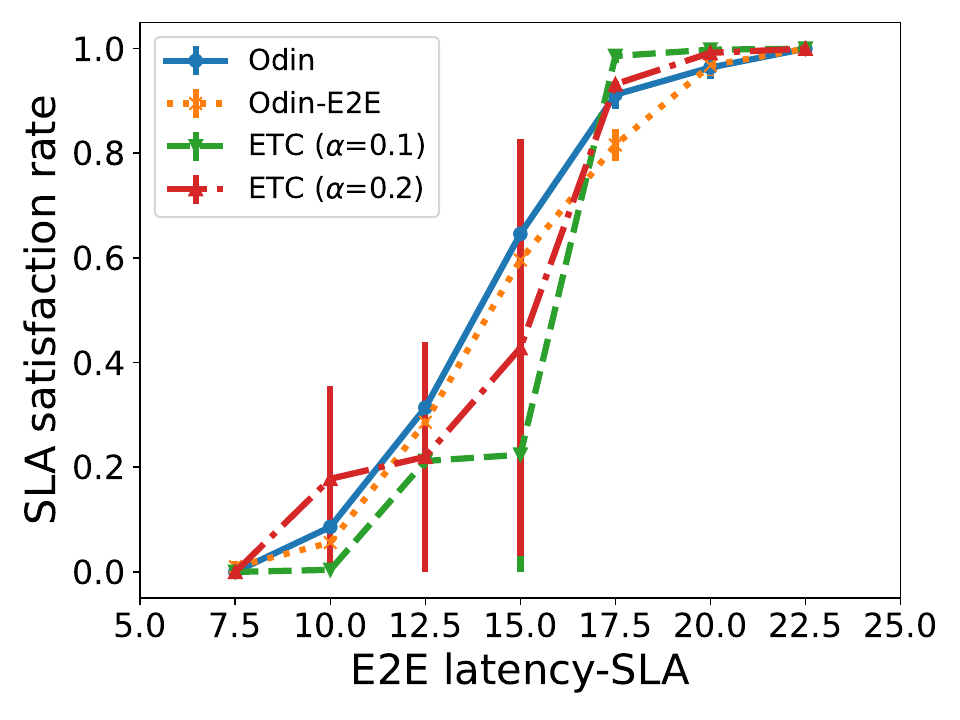}
  }
  \subfloat[E2E latency performance] 
  {\label{fig:syn_avg_lat_low_short}\includegraphics[width=0.25\textwidth]{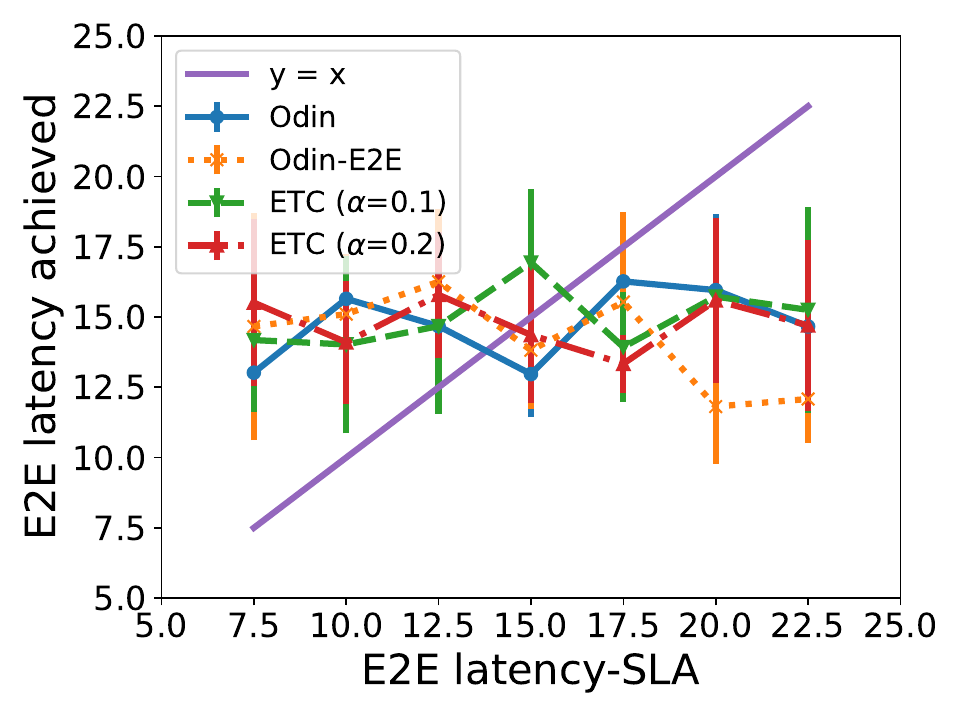}
  }
  \caption{Simulation results for $T=100$ rounds and noise level $R=5\%$} 
  \label{fig:syn_plot_low_short}           
\end{figure*}

\begin{figure*}[t]   
  \centering           
  \subfloat[Regret]
  { \label{fig:syn_reg_high_short}\includegraphics[width=0.25\textwidth]{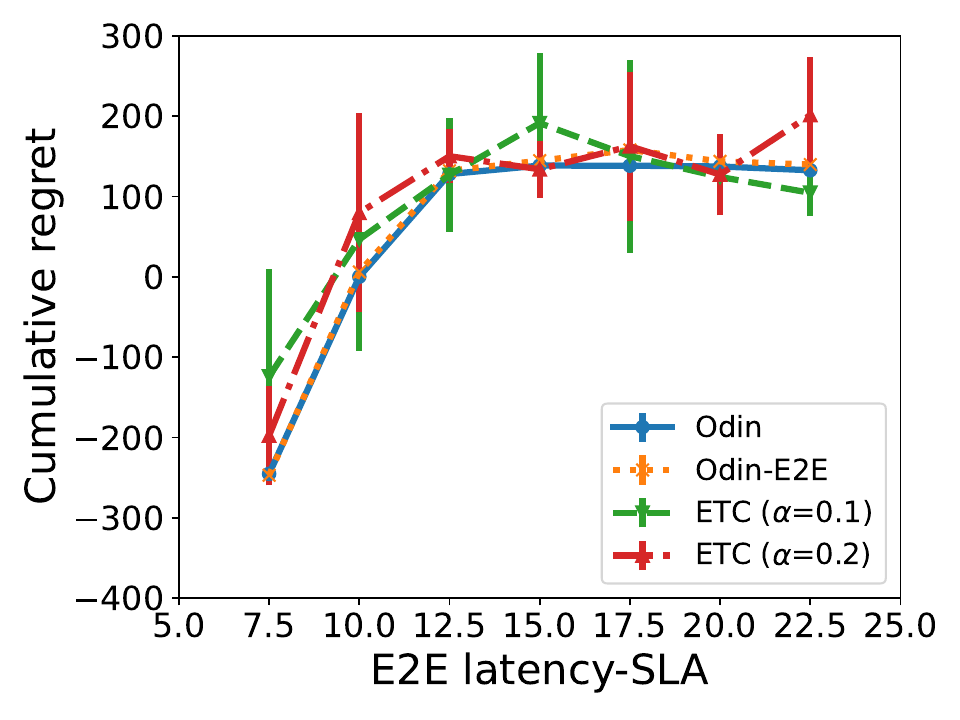}}
  \subfloat[Constraint violation]
{\label{fig:syn_con_vio_high_short}\includegraphics[width=0.25\textwidth]{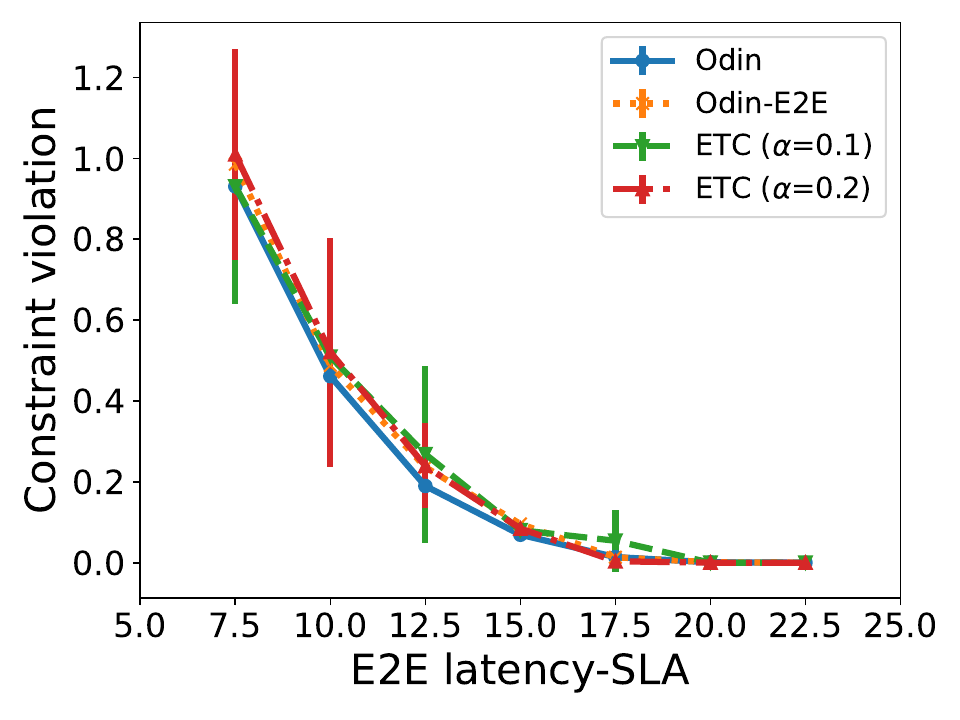}
  }
  \subfloat[No. (\%) rounds with zero violation] 
  {
      \label{fig:syn_satis_high_short}\includegraphics[width=0.25\textwidth]{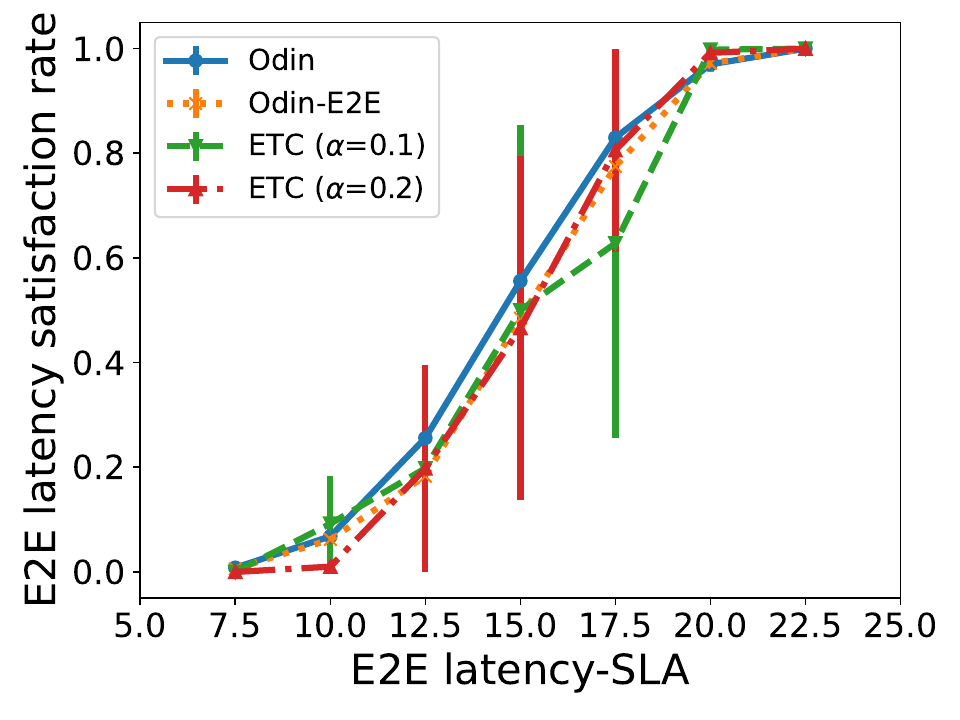}
  }
  \subfloat[E2E latency performance] 
  {\label{fig:syn_avg_lat_high_short}\includegraphics[width=0.25\textwidth]{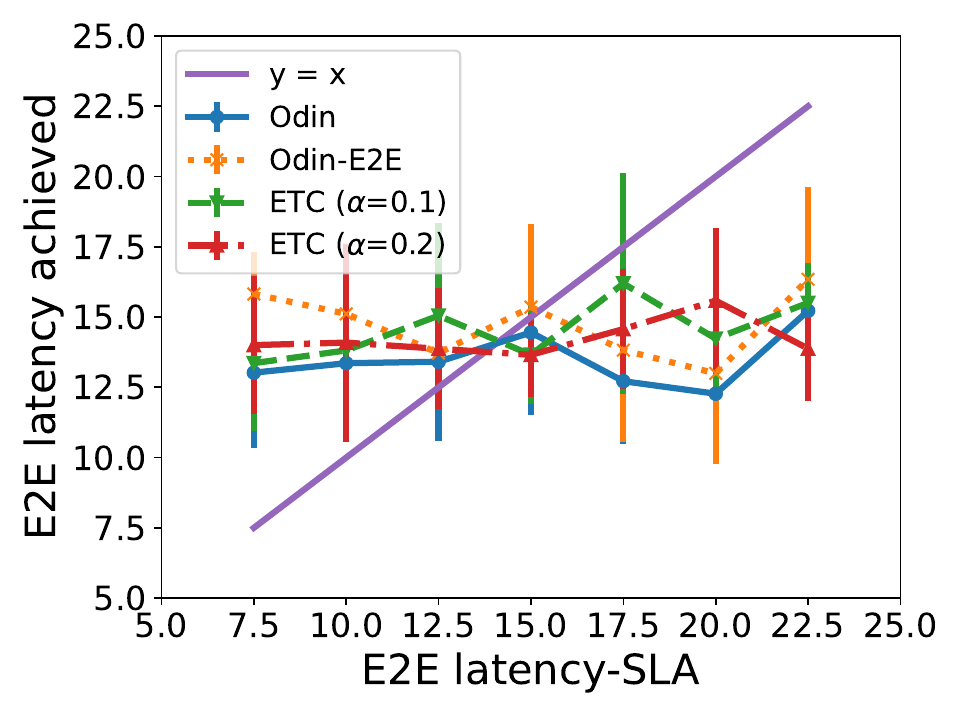}
  }
  \caption{Simulation results for $T=100$ rounds and noise level $R=50\%$} 
  \label{fig:syn_plot_high_short}           
\end{figure*}

\begin{figure*}[t]   
  \centering           
  \subfloat[Regret]
  { \label{fig:syn_reg_low_long}\includegraphics[width=0.25\textwidth]{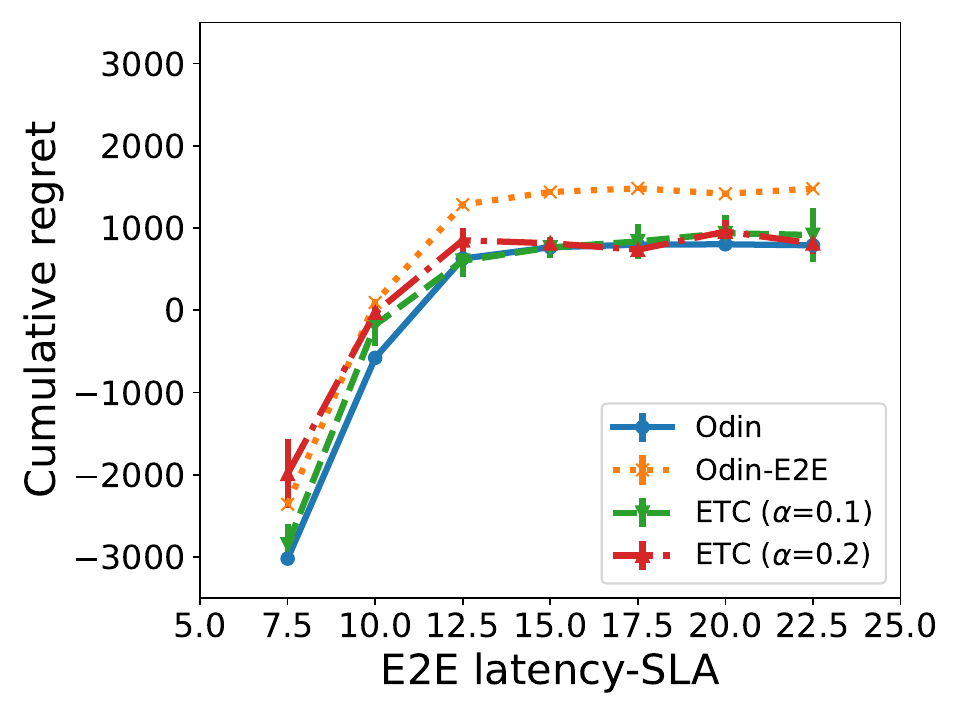}
      }
  \centering           
  \subfloat[Constraint violation]
{\label{fig:syn_con_vio_low_long}\includegraphics[width=0.25\textwidth]{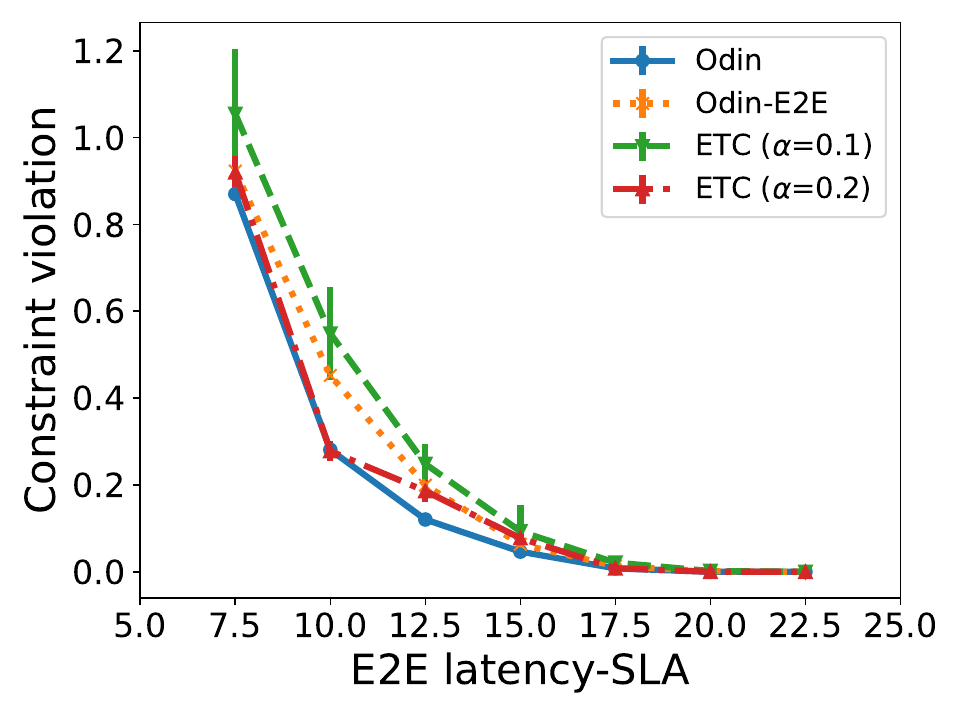}}
  \centering           
  \subfloat[No. (\%) rounds with zero violation] 
{\label{fig:syn_satis_low_long}\includegraphics[width=0.25\textwidth]{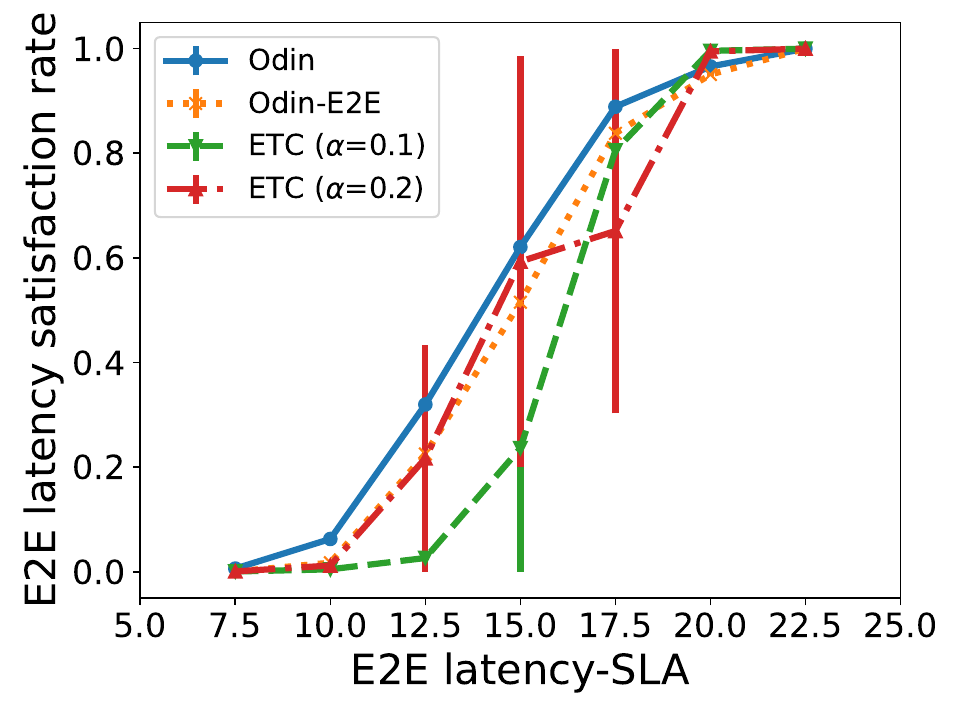}
  }
  \centering           
  \subfloat[E2E latency performance] 
  {\label{fig:syn_avg_lat_low_long}\includegraphics[width=0.25\textwidth]{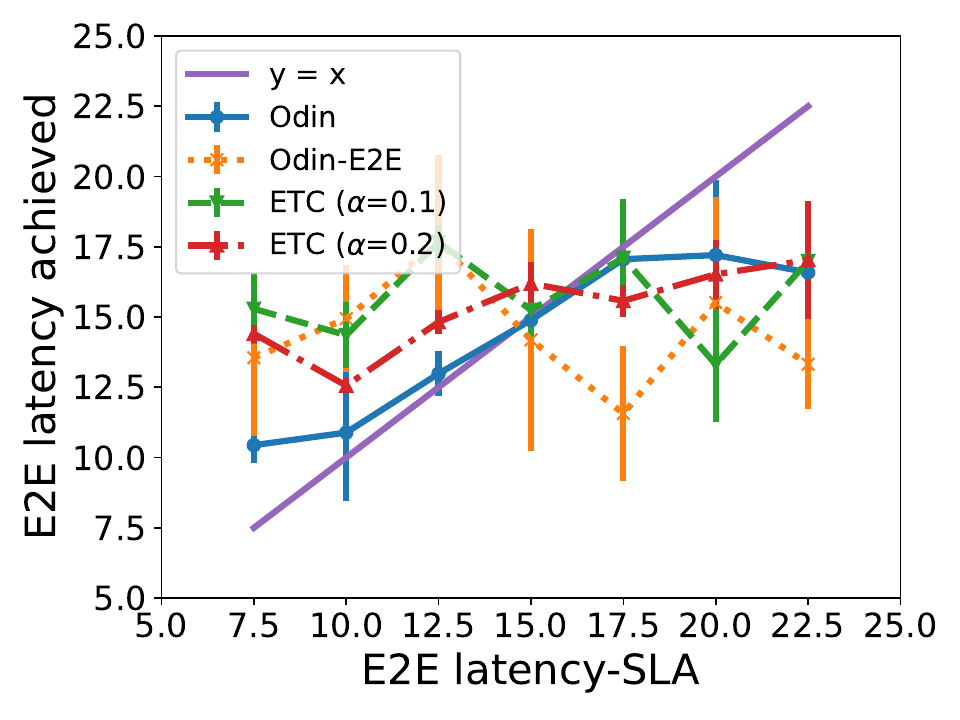}
  }
  \caption{Simulation results for $T=1,000$ rounds and noise level $R=5\%$} 
  \label{fig:syn_plot_low_long}           
\end{figure*}   
\begin{figure*}[t]   
  \centering           
  \subfloat[Regret]
  { \label{fig:syn_reg_high_long}\includegraphics[width=0.25\textwidth]{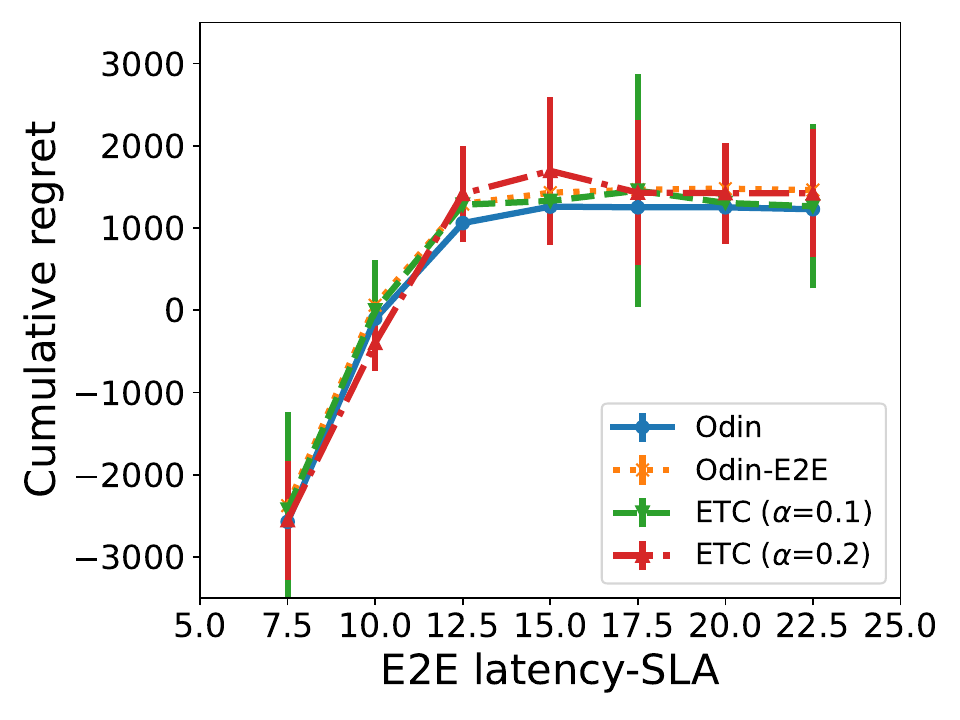}
      }
  \subfloat[Constraint violation]{\label{fig:syn_con_vio_high_long}\includegraphics[width=0.25\textwidth]{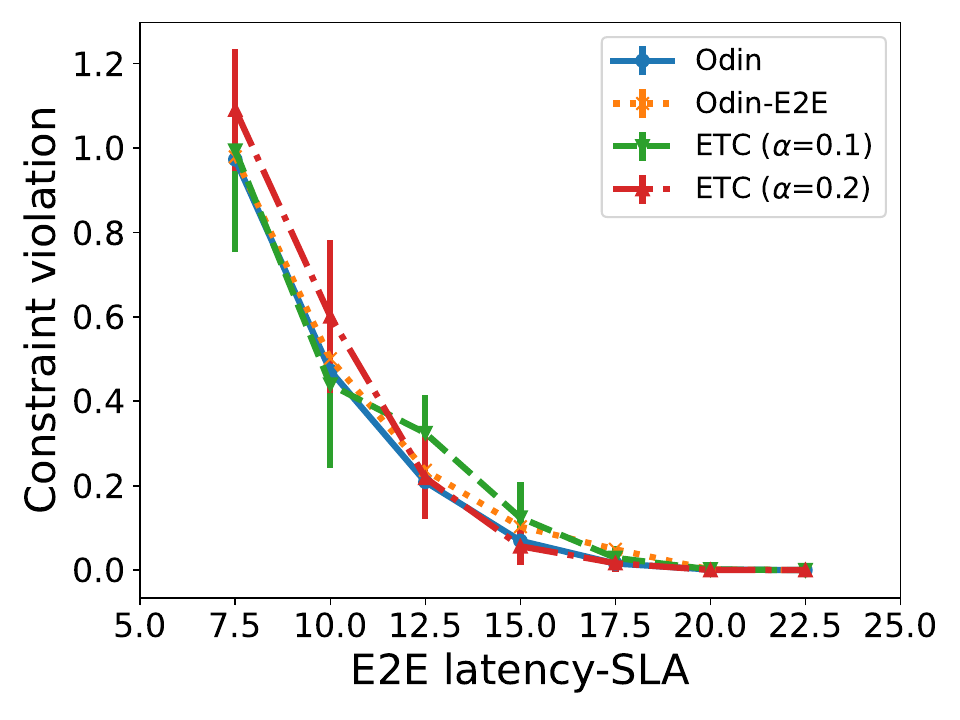}}
  \subfloat[No. (\%) rounds with zero violation] 
  {\label{fig:syn_satis_high_long}\includegraphics[width=0.25\textwidth]{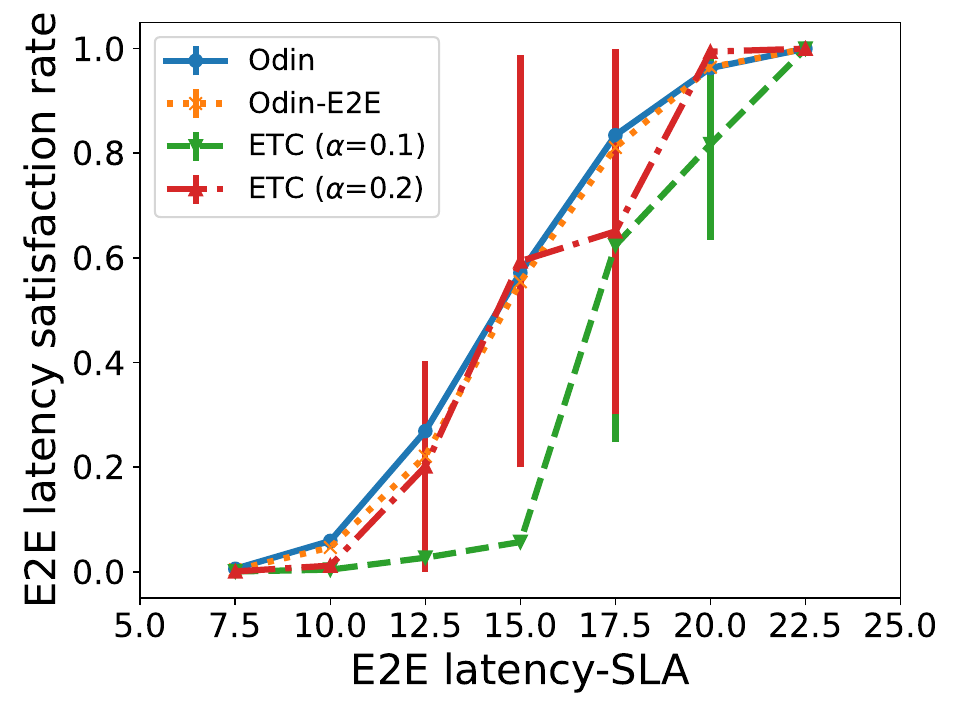}
  }
  \subfloat[E2E latency performance] 
  {\label{fig:syn_avg_lat_high_long}\includegraphics[width=0.25\textwidth]{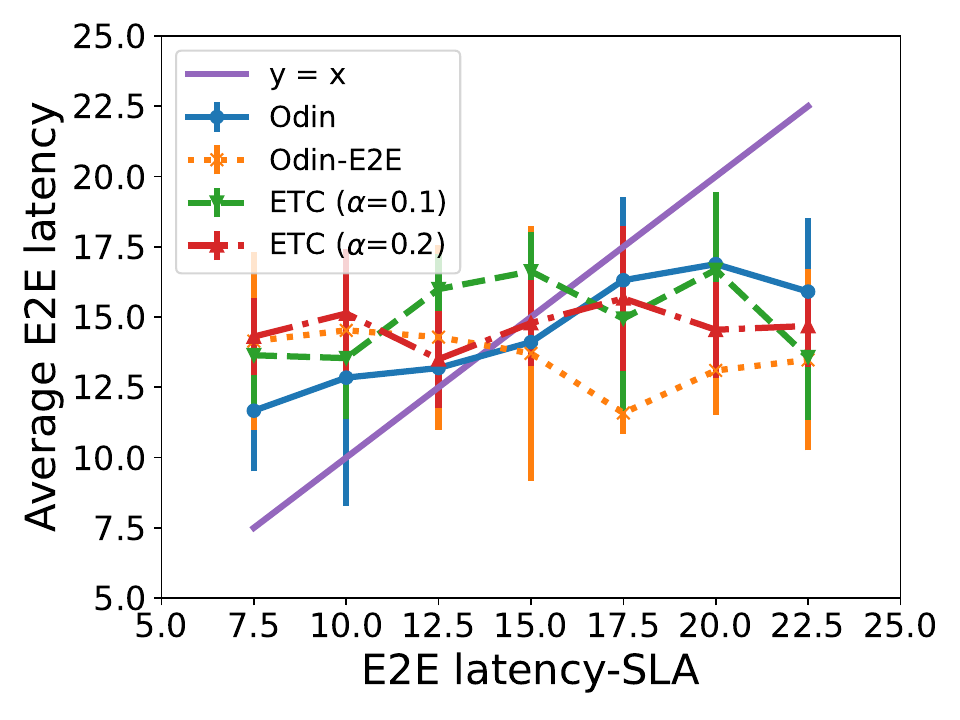}
  }

  \caption{Simulation results for $T=1,000$ rounds and noise level $R=50\%$} 
  \label{fig:syn_plot_high_long}           
\end{figure*}

\subsection{Large-Scale Evaluation - Simulation}
\label{sec:simulations}
We now discuss our simulation results showing \system's efficacy vis-a-vis baseline solutions for disparate complex unknown functions.
To capture heterogeneity across domains, we select diverse functions for each domain. To be specific, 
for the unknown performance (latency) functions ($f^1,f^2,f^3$), we use a logarithmic function of the form $w\log(x)$, exponential 
function of the form $w\exp(x)$, and a quadratic function of the form $wx^2$. Similarly, for cost functions ($g^1, g^2, g^3$), we use a 
rational function of the form $\frac{w}{x+b}$, inversed quadratic function of the form $w/x^2$, and a Gaussian function 
of the form $w\exp(-x^2)$, where $w$ and $b$ are scaling/shifting factors that vary across functions.
In each experiment, the domain function returns a noisy feedback to the
E2E orchestrator, wherein the noise is explicitly added as Gaussian noise. We show the results for two noise levels $R$: $R=5\%$ and $50\%$
of the function's maximum value.

We conduct experiments for various E2E latency target values, i.e., $L\in\{7.5, 10.0,\dots,22.5\}$, which covers the entire range of the achievable E2E latency.
In each experiment,
we conduct 10 independent runs, and for two values of $T=100$ and $1000$, in each run. We plot the 
cumulative regret and 
the constraint violation. The plots in Figs.~\ref{fig:syn_plot_low_short}-\ref{fig:syn_plot_high_long}, show the regret and constraint violation averaged over the 10 runs along with their 
standard deviation. The graphs also show the E2E latency satisfaction rate (i.e., the percentage of total rounds in which the E2E latency is not violated), which is calculated as $\sum_{t=1}^T 1\{\sum_{i=1}^D f^i(x_t^i)\leq L\}/T$ (where $1\{\cdot\}$ is the indicator function), and the $T$-round average E2E latency achieved for each SLA target.

Comparing the results between the two noise levels, 5\% and 50\%, we see that all the solutions, as expected, perform worse when the  noise level is higher. However, for a given
number of rounds $T$, when noise level is at 50\%, \system, in comparison to other baselines, enjoys the lowest regret and the constraint violation,
indicating that it is more robust in the presence of noisy feedback and fewer sampling rounds. Consequently, this results in an increased SLA satisfaction (i.e., more rounds in which the solution has a configuration with 100\% SLA satisfaction), and 
the average E2E latency performance improvement. E.g., for noise level $R= 5\%$, $T=100$, and E2E SLA requirement of 7.5~ms, \system's performance improvement over
the baselines \system-E2E, ETC-0.1, ETC-0.2 is 25\%, 20\%, and 40\%, respectively. However, when the number of observations is increased to $T=1000$, for 
the same SLA = 7.5~ms,
\system improves the E2E latency achieved by 46\%, 60\% and 26.6\%,
compared to the baselines \system-E2E, ETC-0.1, ETC-0.2, respectively. 
When noise level is at 50\%, and $T=100$, the performance improvement for \system over the baselines for SLA = 7.5~ms is 45\%, 5.3\%, and 12\%, respectively. When $T=1000$, even with 50\% noise levels, and for SLA = 7.5~ms, the performance of \system improves by baselines by 26\%, 21\%, and 26.5\%, respectively.
Indeed, for higher SLAs $\geq$ 15~ms, we see that all the algorithms cannot achieve
the latency as demanded, with baselines in some cases achieving latency below \system's, but this is an artifact of \system's optimization which prioritizes lowering 
the cost (resources) as long as the SLA is satisfied, while the other baselines do not necessarily focus on saving costs. In summary, we see \system effectively
balance between performance and cost, while achieving a latency decomposition that can satisfy the SLA. 

\section{Conclusion}
\label{sec:conclusion}
This paper presents \system, a novel Bayesian Optimization-based solution for online E2E SLA decomposition in 5G/6G network slicing. 
\system decomposes a given SLA, assigning specific targets to each domain whilst not needing to know apriori their conditions/control policies. \system's design is fully
compliant with 3GPP standards, UNEXT-native, and can be deployed on existing 5G networks. Our theoretical analysis and trace-based evaluations demonstrate the efficacy and performance gains of \system over baseline solutions. 
\begin{acks}
This work was supported in part by the National Science Foundation under Grants CNS-2312833, CNS-2312835, and CNS-2153220.
Part of this work was done during Duo Cheng's internship at Nokia Bell Labs.
\end{acks}

\bibliographystyle{ACM-Reference-Format}
\bibliography{sample-base.bib}


\end{document}